\documentclass[11pt]{article}
\usepackage[left=1in, right=1in, top=1in, bottom=1in]{geometry}
\usepackage{CJK}
\usepackage{latexsym,bm}
\usepackage{xypic}
\usepackage{amsmath}
\usepackage{wasysym}
\usepackage{indentfirst}
\usepackage{amssymb}
\usepackage{dsfont}
\usepackage{amsthm}
\begin{document}
\newcommand{\fr}[2]{\frac{\;#1\;}{\;#2\;}}
\newtheorem{theorem}{Theorem}[section]
\newtheorem{lemma}{Lemma}[section]
\newtheorem{proposition}{Proposition}[section]
\newtheorem{corollary}{Corollary}[section]
\newtheorem{conjecture}{Conjecture}[section]
\newtheorem{remark}{Remark}[section]
\newtheorem{definition}{Definition}[section]
\newtheorem{example}{Example}[section]
\newtheorem{notation}{Notation}[section]
\numberwithin{equation}{section}
\newcommand{\Aut}{\mathrm{Aut}\,}
\newcommand{\CSupp}{\mathrm{CSupp}\,}
\newcommand{\Supp}{\mathrm{Supp}\,}
\newcommand{\rank}{\mathrm{rank}\,}
\newcommand{\col}{\mathrm{col}\,}
\newcommand{\len}{\mathrm{len}\,}
\newcommand{\leftlen}{\mathrm{leftlen}\,}
\newcommand{\rightlen}{\mathrm{rightlen}\,}
\newcommand{\length}{\mathrm{length}\,}
\newcommand{\bin}{\mathrm{bin}\,}
\newcommand{\wt}{\mathrm{wt}\,}
\newcommand{\diff}{\mathrm{diff}\,}
\newcommand{\lcm}{\mathrm{lcm}\,}
\newcommand{\GL}{\mathrm{GL}\,}
\newcommand{\SJ}{\mathrm{SJ}\,}
\newcommand{\LG}{\mathrm{LG}\,}
\newcommand{\bij}{\mathrm{bij}\,}
\newcommand{\dom}{\mathrm{dom}\,}
\newcommand{\fun}{\mathrm{fun}\,}
\newcommand{\SUPP}{\mathrm{SUPP}\,}
\newcommand{\supp}{\mathrm{supp}\,}
\newcommand{\End}{\mathrm{End}\,}
\newcommand{\Hom}{\mathrm{Hom}\,}
\newcommand{\ran}{\mathrm{ran}\,}
\newcommand{\row}{\mathrm{row}\,}
\newcommand{\Mat}{\mathrm{Mat}\,}
\newcommand{\rk}{\mathrm{rk}\,}
\newcommand{\rs}{\mathrm{rs}\,}
\newcommand{\piv}{\mathrm{piv}\,}
\newcommand{\perm}{\mathrm{perm}\,}
\newcommand{\rsupp}{\mathrm{rsupp}\,}
\newcommand{\inv}{\mathrm{inv}\,}
\newcommand{\orb}{\mathrm{orb}\,}
\newcommand{\id}{\mathrm{id}\,}
\newcommand{\soc}{\mathrm{soc}\,}
\newcommand{\unit}{\mathrm{unit}\,}
\newcommand{\word}{\mathrm{word}\,}

\renewcommand{\thefootnote}{\fnsymbol{footnote}}

\title{$r$-Minimal Codes with Respect to Rank Metric}
\author{Yang Xu$^1$ \,\,\,\,\,\, Haibin Kan$^2$\,\,\,\,\,\,Guangyue Han$^3$}

\maketitle

\renewcommand{\thefootnote}{\fnsymbol{footnote}}

\footnotetext{\hspace*{-6mm} \begin{tabular}{@{}r@{}p{16cm}@{}}
$^1$ & Shanghai Key Laboratory of Intelligent Information Processing, School of Computer Science, Fudan University,
Shanghai 200433, China.\\
&Shanghai Engineering Research Center of Blockchain, Shanghai 200433, China. {E-mail:xuyyang@fudan.edu.cn}\\
$^2$ & Shanghai Key Laboratory of Intelligent Information Processing, School of Computer Science, Fudan University,
Shanghai 200433, China.\\
&Shanghai Engineering Research Center of Blockchain, Shanghai 200433, China.\\
&Yiwu Research Institute of Fudan University, Yiwu City, Zhejiang 322000, China. {E-mail:hbkan@fudan.edu.cn} \\
$^3$ & Department of Mathematics, Faculty of Science, The University of Hong Kong, Pokfulam Road, Hong Kong, China. {E-mail:ghan@hku.hk} \\
\end{tabular}}

\vskip 3mm

{\hspace*{-6mm}{\bf Abstract}---In this paper, we propose and study $r$-minimal codes, a natural extension of minimal codes which have been extensively studied with respect to Hamming metric, rank metric and sum-rank metric. We first propose $r$-minimal codes in a general setting where the ambient space is a finite dimensional left module over a division ring and is supported on a lattice. We characterize minimal subcodes and $r$-minimal codes, derive a general singleton bound, and give existence results for $r$-minimal codes by using combinatorial arguments. We then consider $r$-minimal rank metric codes over a field extension $\mathbb{E}/\mathbb{F}$ of degree $m$, where $\mathbb{E}$ can be infinite. We characterize these codes in terms of cutting $r$-blocking sets, generalized rank weights of the codes and those of the dual codes, and classify codes whose $r$-dimensional subcodes have constant rank support weight. Next, with the help of the evasiveness property of cutting $r$-blocking sets and some upper bounds for the dimensions of evasive subspaces, we derive several lower and upper bounds for the minimal length of $r$-minimal codes. Furthermore, when $\mathbb{E}$ is finite, we establish a general upper bound which generalizes and improves the counterpart for minimal codes in the literature. As a corollary, we show that if $m=3$, then for any $k\geqslant2$, the minimal length of $k$-dimensional minimal codes is equal to $2k$. To the best of our knowledge, when $m\geqslant3$, there was no known explicit formula for the minimal length of $k$-dimensional minimal codes for arbitrary $k$ in the literature.
}\!

\section{Introduction}
\setlength{\parindent}{2em}
A Hamming metric linear code is said to be minimal if all of its non-zero codewords have minimal Hamming supports. In early 1990s, Massey showed in \cite{36} that the minimal codewords (i.e., codewords having minimal Hamming supports) in the dual code completely specify the access structure of the secret sharing scheme based on a linear code. Minimal Hamming metric codes have since been extensively studied due to their applications in secret sharing and secure two-party computation (see \cite{5,16,21,35}) as well as their interesting algebraic, combinatorial and geometric properties (see \cite{1,2,5,6,11,12,22,27,46} and references therein). Among others, it has been shown recently in \cite{1,46} that minimal codes are equivalent to cutting blocking sets, a class of combinatorial objects introduced in \cite{12} for constructing minimal codes. Many important results for minimal codes have since been established by using cutting blocking sets. For example, Alfarano, Borello and Neri proved in \cite{1} that minimal codes are asymptotically good; later in \cite{5}, Alon, Bishnoi, Das and Neri gave explicit constructions of asymptotically good families of minimal codes; in \cite{27}, H\'{e}ger and Nagy established upper bounds for minimal length of minimal codes which are linear in both the dimension and the field size; later in \cite{11}, Bishnoi, D'haeseleer, Gijswijt and Potukuchi further improved the counterpart result in \cite{27} to more general settings. More recently in \cite{4}, Alfarano, Borello and Neri introduced and investigated the notion of outer minimal codes and outer strong blocking sets, two larger class of objects which are closely related to minimal codes and cutting blocking sets, respectively (see \cite{4} for more details).

Recently, there has been a growing interest in minimal rank metric codes (see Section 2.1 for more details), which were first introduced by Alfarano, Borello, Neri and Ravagnani in \cite{3}. It was shown in \cite{3} that these codes are equivalent to linear cutting blocking sets, a class of combinatorial objects which are both special cases and $q$-analogues of cutting blocking sets. In \cite{3}, the authors established bounds for the parameters of minimal rank metric codes, gave general existence results via a combinatorial argument, and gave explicit constructions for some classes of minimal rank metric codes; moreover, they showed that a minimal rank metric code can be associated to a minimal Hamming metric code (see [3, Sections 5,6] for more details). Many important results for minimal codes have since been established by using linear cutting blocking sets. In \cite{8}, Bartoli, Csajb\'{o}k, Marino and Trombetti proved that linear cutting blocking sets are equivalent to a class of evasive subspaces (see \cite{9,33}), which further led to the genuinely new characterization of minimal rank metric codes in terms of their second generalized rank weights; moreover, they constructed $4$-dimensional minimal codes over $\mathbb{F}_{q^{4}}/\mathbb{F}_{q}$ with the shortest length $8$, which also provided the first example of a class of linear MRD codes that are not obtained as a direct sum of two smaller MRD codes (see [8, Sections 3,4] for more details). In \cite{31}, Lia, Longobardi, Marino and Trombetti generalized the evasiveness property of linear cutting blocking sets to linear cutting $r$-blocking sets, and constructed $3$-dimensional minimal codes over $\mathbb{F}_{q^{m}}/\mathbb{F}_{q}$ with the shortest length $m+2$ for all odd $m\geqslant5$ (see [31, Sections 3,4] for more details). More recently, the notion of minimal codes has been further generalized to sum-rank metric \cite{34}, a metric that includes both Hamming metric and rank metric as special cases; we refer the reader to \cite{13,38,44} for more details.

In this paper, we propose and study $r$-minimal codes for $r\in\mathbb{N}$. Roughly speaking, $r$-minimal codes are natural extensions of minimal codes with codewords replaced by $r$-dimensional subcodes in the definition, and boil down to minimal codes when $r=1$. We first study $r$-minimal codes in a general setting, and then restrict our attention to rank metric. Other than being an extension of minimal codes, $r$-minimal rank metric codes can be regarded as the coding-theoretic counterpart of linear cutting $r$-blocking sets (see \cite{12,26,31}). This generalizes the equivalence between minimal rank metric codes and linear cutting blocking sets.

The main contributions of this paper can be summarized as follows.

In Section 3, we propose and study $(\sigma,r)$-minimal codes as submodules of a finite dimensional left $R$-module $X$, where $R$ is a possibly infinite division ring, and $X$ is supported on a lattice $Y$ with respect to $\sigma:X\longrightarrow Y$. More precisely, in Section 3.1, with $X$, $Y$ and $\sigma$ appropriately set, we give six examples of $r$-minimal codes with respect to various metrics. In Section 3.2, we characterize minimal subcodes of a given code and $(\sigma,r)$-minimal codes (Theorems 3.2 and 3.3), and derive several corollaries including a general Singleton bound and the fact that $(\sigma,r)$-minimal codes are also $(\sigma,s)$-minimal for $0\leqslant s\leqslant r$ (Corollaries 3.1--3.3). In Section 3.3,  when $R$ is finite, we derive two existence results for $(\sigma,r)$-minimal codes of a given dimension by using combinatorial arguments (Theorems 3.4 and 3.5).

In Sections 4 and 5, we study $r$-minimal rank metric codes over a field extension $\mathbb{E}/\mathbb{F}$ of degree $m$, where $\mathbb{E}$ can be infinite.

Section 4 is devoted to the characterizations and basic properties of $r$-minimal codes and linear cutting $r$-blocking sets. In Section 4.1, we derive the maximal rank support weight of all the $s$-dimensional subcodes of a code $C$ for all $0\leqslant s\leqslant\dim_{\mathbb{E}}(C)$ (Theorem 4.1). This result extends [10, Theorem 3.1] and is crucial for our discussion in later sections. In Section 4.2, we establish the evasiveness property of linear cutting $r$-blocking sets (Theorem 4.2) and derive some corollaries. In Section 4.3, we characterize $r$-minimal codes in terms of linear cutting $r$-blocking sets, generalized rank weights of the codes and those of the dual codes (Theorems 4.3--4.5), where the equivalence between $r$-minimal codes and cutting $r$-blocking sets also extends to sum-rank metric (Remark 4.3); moreover, we characterize rank metric codes whose $r$-dimensional subcodes have constant rank support weight (Theorem 4.6). Some of our results generalize the counterpart results in \cite{3,8,31,39} to arbitrary $r$ and possibly infinite $\mathbb{E}$ via a different approach.

Section 5 is devoted to the minimal length $\varpi_{\mathbb{E}/\mathbb{F}}(k,r)$ of $k$-dimensional $r$-minimal codes for $k\geqslant r+1$, which turns out to be equal to the minimal dimension of linear cutting $r$-blocking sets of $\mathbb{E}^{k}$ (Proposition 5.1, Corollary 5.1). In Section 5.1, we derive some upper bounds for the dimensions of evasive subspaces (Theorem 5.1, Corollary 5.2). In Section 5.2, we establish several lower bounds for $\varpi_{\mathbb{E}/\mathbb{F}}(k,r)$ with the help of the results established in Section 5.1 (Theorems 5.2 and 5.3). In Section 5.3, when $\mathbb{E}$ is finite, we derive the explicit number of all the $[n,r+1]$ $r$-minimal codes (Proposition 5.3), and establish in Theorem 5.4 the general upper bound
$$\varpi_{\mathbb{E}/\mathbb{F}}(k,r)\leqslant mr+k(r+1)-r^{2}-2r.$$
As a special case, the upper bound $m+2k-3$ for minimal codes slightly improves [3, Theorem 6.11]. Furthermore, we show that $\varpi_{\mathbb{E}/\mathbb{F}}(k,1)=2k$ if $m=3$. To the best of our knowledge, when $m\geqslant3$, there was no known explicit formula for $\varpi_{\mathbb{E}/\mathbb{F}}(k,1)$ for arbitrary $k$ in the literature.

\section{Preliminaries}
\setlength{\parindent}{2em}
We begin with a few notations for modules over a division ring $R$. For a left $R$-module $X$ and any $A\subseteq X$, let $\langle A\rangle_{R}$ denote the $R$-submodule of $X$ generated by $A$, and write $A\leqslant_{R}X$ if $A$ is an $R$-submodule of $X$. For $m,n\in\mathbb{Z}^{+}$, let $R^{n}$ and $R^{[n]}$ denote the sets of all the row vectors and column vectors over $R$ of length $n$, respectively, and let $\Mat_{m,n}(R)$ denote the set of all the matrices over $R$ with $m$ rows and $n$ columns; moreover, for any $G\in\Mat_{m,n}(R)$, let $\row(G)$ and $\col(G)$ denote the sets of all the rows and columns of $G$, respectively.

\subsection{Rank metric codes}
\setlength{\parindent}{2em}
Rank metric codes were introduced by Delsarte in \cite{21} and by Gabidulin in \cite{24}. The study of rank metric codes has intensified greatly in the last decades due to their applications in secure network coding and crisscross error correction (see \cite{30,35,42,43}) as well as their important mathematical properties (see \cite{18,23,25,39,40} and references therein) that may hold true for codes over possibly infinite fields (see, e.g., \cite{7,10,29,34,35}). Among others, Berhuy, Fasel and Garotta proved in \cite{10} that if a rank metric code is of length less than or equal to the field extension degree, then there exists a codeword whose rank support is equal to the rank support of the whole code (see [10, Theorem 3.1]), settling a conjecture proposed by Jurrius and Pellikaan in \cite{29}. This result is a keystone for our approach as it enables us to establish many of our results for codes over possibly infinite fields (see Remark 2.3 for more details).

Throughout this subsection, let $\mathbb{E}/\mathbb{F}$ be a finite dimensional field extension with $\dim_{\mathbb{F}}(\mathbb{E})=m$. For any $s\in\mathbb{Z}^{+}$ and $A\subseteq \mathbb{E}^{s}$, the dual of $A$, denoted by $A^{\bot}$, is defined as
\begin{equation}A^{\bot}\triangleq\{\beta\in \mathbb{E}^{s}\mid \alpha\beta^{T}=0~\text{for all $\alpha\in A$}\},\end{equation}
where $\beta^{T}\in \mathbb{E}^{[s]}$ denotes the transport of $\beta\in \mathbb{E}^{s}$; moreover, let
\begin{equation}A^{\ddagger}\triangleq\{\beta^{T}\mid \beta\in A^{\bot}\}\subseteq\mathbb{E}^{[s]}.\end{equation}
We also fix $n\in\mathbb{Z}^{+}$. For any $k\in\mathbb{N}$, an \textit{$[n,k]$ rank metric code} is a $k$-dimensional $\mathbb{E}$-subspace of $\mathbb{E}^{n}$. For $k\in\mathbb{Z}^{+}$ and an $[n,k]$ rank metric code $C\leqslant_{\mathbb{E}}\mathbb{E}^{n}$, a \textit{generator matrix of $C$} is a matrix $G\in\Mat_{k,n}(\mathbb{E})$ satisfying $C=\langle\row(G)\rangle_{\mathbb{E}}=\{\gamma G\mid \gamma\in\mathbb{E}^{k}\}$.

Now let $(\tau_1,\dots,\tau_m)$ be an ordered basis of $\mathbb{E}/\mathbb{F}$. Then, for any $\alpha\in\mathbb{E}^{n}$, there uniquely exists $\mathcal{M}(\alpha)\in\Mat_{m,n}(\mathbb{F})$ such that $\alpha=(\tau_1,\dots,\tau_m)\mathcal{M}(\alpha)$, and the \textit{rank support of $\alpha$}, denoted by $\rsupp(\alpha)$, is defined as the $\mathbb{F}$-subspace of $\mathbb{F}^{n}$ generated by $\row(\mathcal{M}(\alpha))$. More generally, for any $A\subseteq\mathbb{E}^{n}$, the \textit{rank support of $A$}, denoted by $\chi(A)$, is defined as the $\mathbb{F}$-subspace of $\mathbb{F}^{n}$ generated by $\bigcup_{\alpha\in A}\rsupp(\alpha)$; moreover, the \textit{rank support weight of $A$}, denoted by $\wt(A)$, is defined as
\begin{equation}\wt(A)\triangleq\dim_{\mathbb{F}}(\chi(A)).\end{equation}
We remark that the rank support of an element is independent of the choice of the ordered basis.

The following lemma, in which we establish some connections between rank supports and dual spaces, will be used frequently in the rest of the paper.

\setlength{\parindent}{0em}
\begin{lemma}
{\bf{(1)}}\,\,For $A\subseteq\mathbb{E}^{n}$, we have $A^{\bot}\cap\mathbb{F}^{n}=\chi(A)^{\bot}\cap\mathbb{F}^{n}$ and $\wt(A)=n-\dim_{\mathbb{F}}(A^{\bot}\cap\mathbb{F}^{n})$.

{\bf{(2)}}\,\,Let $k\in\mathbb{Z}^{+}$, and let $G\in\Mat_{k,n}(\mathbb{E})$, $C=\langle\row(G)\rangle_{\mathbb{E}}$, $U=\langle\col(G)\rangle_{\mathbb{F}}$. Define $\varphi:\mathbb{F}^{n}\longrightarrow\mathbb{E}^{[k]}$ as $\varphi(\theta)=G\theta^{T}$. Then, it holds that $\ran(\varphi)=U$, $\ker(\varphi)=C^{\bot}\cap\mathbb{F}^{n}$ and $\wt(C)=\dim_{\mathbb{F}}(U)$. Moreover, for $B,P\subseteq\mathbb{E}^{k}$ and $D\triangleq\{\gamma G\mid\gamma\in B\}$, $Q\triangleq\{\gamma G\mid\gamma\in P\}$, we have $\varphi^{-1}[{B^{\ddagger}}\cap U]=D^{\bot}\cap\mathbb{F}^{n}$, $\wt(D)=\dim_{\mathbb{F}}(U)-\dim_{\mathbb{F}}({B^{\ddagger}}\cap U)$, and it holds that $\chi(Q)\subseteq\chi(D)\Longleftrightarrow{B^{\ddagger}}\cap U\subseteq P^{\ddagger}$.
\end{lemma}

\begin{proof}
{\bf{(1)}}\,\,For $\beta\in \mathbb{F}^{n}$ and any $\alpha\in A$, we have
$$\alpha\beta^{T}=0\Longleftrightarrow(\tau_1,\dots,\tau_m)\mathcal{M}(\alpha)\beta^{T}=0\Longleftrightarrow\mathcal{M}(\alpha)\beta^{T}=0\Longleftrightarrow\beta\in\rsupp(\alpha)^{\bot},$$
which further implies that $\beta\in A^{\bot}\Longleftrightarrow\beta\in\chi(A)^{\bot}$, as desired. Moreover, from $\chi(A)\leqslant_{\mathbb{F}}\mathbb{F}^{n}$ and (2.3), we deduce that $\wt(A)=n-\dim_{\mathbb{F}}(\chi(A)^{\bot}\cap\mathbb{F}^{n})=n-\dim_{\mathbb{F}}(A^{\bot}\cap\mathbb{F}^{n})$, as desired.

{\bf{(2)}}\,\,First of all, it is straightforward to verify that $\ran(\varphi)=U$ and $\ker(\varphi)=C^{\bot}\cap\mathbb{F}^{n}$. It then follows from (1) that $\wt(C)=n-\dim_{\mathbb{F}}(\ker(\varphi))=\dim_{\mathbb{F}}(U)$, as desired. Next, for $\theta\in\mathbb{F}^{n}$, we have
$$\theta\in\varphi^{-1}[{B^{\ddagger}}\cap U]\Longleftrightarrow G\theta^{T}\in B^{\ddagger}\Longleftrightarrow(\forall~\gamma\in B:\gamma G\theta^{T}=0)\Longleftrightarrow\theta\in D^{\bot},$$
which implies that $\varphi^{-1}[{B^{\ddagger}}\cap U]=D^{\bot}\cap\mathbb{F}^{n}$, as desired. It then follows from (1) that
$$\wt(D)=n-\dim_{\mathbb{F}}(\varphi^{-1}[{B^{\ddagger}}\cap U])=n-(\dim_{\mathbb{F}}(\ker(\varphi))+\dim_{\mathbb{F}}({B^{\ddagger}}\cap U))=\dim_{\mathbb{F}}(U)-\dim_{\mathbb{F}}({B^{\ddagger}}\cap U),$$
as desired. Finally, it follows from $\chi(D),\chi(Q)\leqslant_{\mathbb{F}}\mathbb{F}^{n}$ and (1) that
$$\chi(Q)\subseteq\chi(D)\Longleftrightarrow D^{\bot}\cap\mathbb{F}^{n}\subseteq Q^{\bot}\cap\mathbb{F}^{n}\Longleftrightarrow\varphi^{-1}[{B^{\ddagger}}\cap U]\subseteq\varphi^{-1}[{P^{\ddagger}}\cap U]\Longleftrightarrow{B^{\ddagger}}\cap U\subseteq{P^{\ddagger}}\cap U,$$
completing the proof.
\end{proof}

\begin{remark}
Lemma 2.1 recovers [3, Lemma 3.7] and [3, Theorem 5.6] if both $B$ and $Q$ contain exactly one codeword. Moreover, since Lemma 2.1 does not require that $\rank(G)=k$, it can be used to derive a generalization of [38, Theorem 3.1] which was established for sum-rank metric codes.
\end{remark}

\setlength{\parindent}{2em}
The following result has been established in [39, Theorem 10] and [3, Proposition 3.4]. A proof is included since it is a direct corollary of Lemma 2.1, and will be used in Section 4.3 when we characterize codes whose subcodes of a given dimension have constant rank support weight.

\setlength{\parindent}{0em}
\begin{corollary}
Let $C\leqslant_{\mathbb{E}}\mathbb{E}^{n}$ with $\dim_{\mathbb{E}}(C)=k\geqslant1$, $G\in\Mat_{k,n}(\mathbb{E})$ be a generator matrix of $C$, and let $U=\langle\col(G)\rangle_{\mathbb{F}}$. Then, we have $\wt(C)\leqslant \dim_{\mathbb{F}}(C)$, where the equality holds if and only if $U=\mathbb{E}^{[k]}$. Moreover, if $\wt(C)=\dim_{\mathbb{F}}(C)$, then we have $\wt(D)=\dim_{\mathbb{F}}(D)$ for all $D\leqslant_{\mathbb{E}}C$.
\end{corollary}

\begin{proof}
By Lemma 2.1, we have $\wt(C)=\dim_{\mathbb{F}}(U)$. This, together with $\dim_{\mathbb{F}}(C)=mk=\dim_{\mathbb{F}}(\mathbb{E}^{[k]})$ and $U\leqslant_{\mathbb{F}}\mathbb{E}^{[k]}$, immediately implies the first assertion. Next, suppose that $\wt(C)=\dim_{\mathbb{F}}(C)$. Then, we have $U=\mathbb{E}^{[k]}$. Consider $D\leqslant_{\mathbb{E}}C$, and let $B\leqslant_{\mathbb{E}}\mathbb{E}^{k}$ such that $D=\{\gamma G\mid\gamma\in B\}$. From Lemma 2.1, we deduce that $\wt(D)=mk-\dim_{\mathbb{F}}(B^{\ddagger})=\dim_{\mathbb{F}}(B)=\dim_{\mathbb{F}}(D)$, as desired.
\end{proof}

\setlength{\parindent}{2em}
Now we define $r$-minimal rank metric codes as a natural extension of minimal rank metric codes (see, e.g., [3, Definition 5.1]).

\setlength{\parindent}{0em}
\begin{definition}
Let $C\leqslant_{\mathbb{E}}\mathbb{E}^{n}$. An $\mathbb{E}$-subspace $D$ of $C$ is said to be rank minimal in $C$ if for any $B\leqslant_{\mathbb{E}}C$ such that $\dim_{\mathbb{E}}(B)=\dim_{\mathbb{E}}(D)$, $\chi(B)\subseteq\chi(D)$, it holds that $\chi(B)=\chi(D)$. For any $r\in\mathbb{N}$, we say that $C$ is $r$-minimal with respect to rank metric if every $r$-dimensional $\mathbb{E}$-subspace of $C$ is rank minimal in $C$. Moreover, we will simply say that $C$ is minimal if $C$ is $1$-minimal.
\end{definition}

\begin{remark}
When $r=1$, Definition 2.1 boils down to [3, Definition 5.1] as minimal codewords and minimal codes proposed there coincide with $1$-dimensional rank minimal subcodes and $1$-minimal codes proposed in Definition 2.1, respectively.
\end{remark}

\setlength{\parindent}{2em}
Now we recall the definition of generalized rank weights. We note that there are several equivalent definitions for generalized rank weights in the literature (see, e.g., \cite{10,23,25,29,39,40}).

\setlength{\parindent}{0em}
\begin{definition}([29, Definition 2.5]) Let $C\leqslant_{\mathbb{E}}\mathbb{E}^{n}$ with $\dim_{\mathbb{E}}(C)=k$. For any $r\in\{0,1,\dots,k\}$, the $r$-th generalized rank weight of $C$ is defined as $\mathbf{d}_{r}(C)\triangleq\min\{\wt(D)\mid D\leqslant_{\mathbb{E}}C,\dim_{\mathbb{E}}(D)=r\}$.
\end{definition}

\setlength{\parindent}{2em}
We end this subsection by introducing a lemma which is the keystone for our approach as it enables us to establish many of our results for possibly infinite $\mathbb{E}$. The following result was first proved by Jurrius and Pellikaan for finite field alphabet in \cite{29}, where the authors also raised the question that whether the statement remains valid for infinite field. Later in \cite{10}, Berhuy, Fasel and Garotta positively answered the question by using methods from algebraic geometry.

\setlength{\parindent}{0em}
\begin{lemma}([10, Theorem 3.1])
If $n\leqslant m$, then for any $C\leqslant_{\mathbb{E}}\mathbb{E}^{n}$, there exists $\alpha\in C$ such that $\rsupp(\alpha)=\chi(C)$.
\end{lemma}

\setlength{\parindent}{2em}
\begin{remark}
When $\mathbb{E}$ is finite, counting methods are extensively used to establish various results for rank metric codes. For example, Lemma 2.2 was established in [29, Proposition 10] by computing the cardinality of the union of a class of subcodes of $C$; in [39, Theorem 12], the authors classified the constant weight codes by using the value function proposed in \cite{31}; in \cite{3}, the authors proved that for a code $C\neq\{0\}$, the largest weight of all the codewords of $C$ is equal to $\min\{m,\wt(C)\}$. Further improving [29, Proposition 10], this result was established by using the standard equations (see [3, Lemma 3.6]), which were used multiple times in \cite{3} to establish various results. As these tools do not directly apply to infinite $\mathbb{E}$, Lemma 2.2 is a perfect substitute for our discussion and enables us to generalize all the above mentioned results to possibly infinite $\mathbb{E}$ via a unified approach.
\end{remark}

\subsection{Evasive subspaces and cutting $r$-blocking sets}
\setlength{\parindent}{2em}
Evasive subspaces have been recently studied in \cite{8,9,33} as a natural generalization of $h$-scattered subspaces (see \cite{19}). It has been shown in [33, Theorem 3.3] that evasive subspaces are deeply connected with the generalized rank weights, which further provide a powerful tool for exploiting the interplay between geometry and coding theory (see \cite{8,33,48} for more details). We refer the reader to \cite{26,44} for generalizations of evasive subspaces and their related codes.

Let $\mathbb{E}/\mathbb{F}$ be a finite dimensional field extension with $\dim_{\mathbb{F}}(\mathbb{E})=m$, and let $X$ be a $k$-dimensional vector space over $\mathbb{E}$. We begin by recalling the definition of evasive subspaces, which has originally been proposed for vector spaces over finite fields and naturally extends to those over infinite fields.

\setlength{\parindent}{0em}
\begin{definition}
For $J\leqslant_{\mathbb{F}}X$ and $(h,r)\in\mathbb{N}\times\mathbb{Z}$, we say that $J$ is $(h,r)$-evasive in $X$ if $\langle J\rangle_{\mathbb{E}}=X$, and for any $M\leqslant_{\mathbb{E}}X$ with $\dim_{\mathbb{E}}(M)=h$, it holds that $\dim_{\mathbb{F}}(J\cap M)\leqslant r$.
\end{definition}

\setlength{\parindent}{2em}
Now we collect some properties of evasive spaces in the following lemma.

\setlength{\parindent}{0em}
\begin{lemma}
{\bf{(1)}}\,\,Let $(h,r)\in\mathbb{N}\times\mathbb{Z}$ with $h\leqslant k$, and let $J\leqslant_{\mathbb{F}}X$ be $(h,r)$-evasive in $X$. Then, it holds that $h\leqslant r$. Moreover, $J$ is $(h-t,r-t)$-evasive in $X$ for all $t\in\{0,1,\dots,h\}$.

{\bf{(2)}}\,\,Let $h\in\{0,1,\dots,k\}$, and let $J\leqslant_{\mathbb{F}}X$ be $(h,h)$-evasive in $X$ with $\dim_{\mathbb{F}}(J)\geqslant k+1$. Then, it holds that $\dim_{\mathbb{F}}(J)\leqslant km/(h+1)$.

{\bf{(3)}}\,\,Let $(b,w)\in\mathbb{N}^{2}$ with $b\leqslant k$, $U\leqslant_{\mathbb{F}}X$ be $(b,w)$-evasive in $X$, $A\leqslant_{\mathbb{E}}X$ such that $\dim_{\mathbb{E}}(A)\triangleq a\leqslant b$, and write $v=\dim_{\mathbb{F}}(U\cap A)$. Then, $(U+A)/A$ is $(b-a,w-v)$-evasive in $X/A$.
\end{lemma}

\begin{proof}
Part (1) follows from [9, Proposition 2.6], and (2) follows from [19, Theorem 2.3] and [33, Corollary 4.4]. We note that although these results were established for vector spaces over finite fields, the proofs in \cite{9,33} directly apply to any infinite $\mathbb{E}$. Therefore we only prove (3). First, it follows from $\langle U\rangle_{\mathbb{E}}=X$ that $\langle (U+A)/A\rangle_{\mathbb{E}}=X/A$. Next, let $T\leqslant_{\mathbb{E}}X/A$ with $\dim_{\mathbb{E}}(T)=b-a$. Then, there exists $W\leqslant_{\mathbb{E}}X$ such that $A\subseteq W$, $T=W/A$. Since $\dim_{\mathbb{E}}(W)=b$ and $U$ is $(b,w)$-evasive in $X$, we have $\dim_{\mathbb{F}}(U\cap W)\leqslant w$, which, together with $((U+A)/A)\cap T=(A+(U\cap W))/A$, implies that
$$\dim_{\mathbb{F}}\left(T\cap\left((U+A)/A\right)\right)=\dim_{\mathbb{F}}(U\cap W)-\dim_{\mathbb{F}}((U\cap W)\cap A)\leqslant w-\dim_{\mathbb{F}}(U\cap A)=w-v,$$
which further establishes the desired result.
\end{proof}

\setlength{\parindent}{2em}
Now we consider cutting $r$-blocking sets, which were introduced by Bonini and Borello in \cite{12}, and boil down to cutting blocking sets when $r=1$ (see [12, Definitions 3.2 and 3.4]). In \cite{12}, the authors used cutting blocking sets to construct minimal Hamming metric codes which do not satisfy the Ashikhmin-Barg condition (see \cite{6}). Later, it has been proven in \cite{1,46} that a Hamming metric linear code is minimal if and only if all the columns of its generator matrix form a cutting blocking set. Such a connection between cutting blocking sets and minimal codes has been extended to rank metric codes in \cite{3}, and to sum-rank metric codes in \cite{13,44}, respectively. In Sections 3.1 and 4.3, we will further generalize this connection to cutting $r$-blocking sets and $r$-minimal codes with respect to various metrics. We also refer the reader to \cite{11,26,31} for more recent results on cutting $r$-blocking sets.

Let $R$ be a possibly infinite division ring, and let $L$ be a $k$-dimensional left $R$-module. We begin by recalling the definition of cutting $r$-blocking sets. This definition has originally been proposed for vector spaces over finite fields, and we naturally extend it to modules over division rings.

\begin{definition}(see [1, Definition 3.2, Proposition 3.3])
For $S\subseteq L$ and $r\in\mathbb{N}$, we say that $S$ is a cutting $r$-blocking set of $L$ if for any $V\leqslant_{R}L$ with $\dim_{R}(V)=k-r$, it holds that $\langle S\cap V\rangle_{R}=V$. A cutting $1$-blocking set is simply referred to as a cutting blocking set.
\end{definition}

\setlength{\parindent}{2em}
The following characterization of cutting $r$-blocking sets will be used in Section 4.2.

\setlength{\parindent}{0em}
\begin{proposition}
Let $S\subseteq L$ with $0\in S$, and let $r\in\{0,1,\dots,k-1\}$. Then, $S$ is a cutting $r$-blocking set of $L$ if and only if for any $W,I\leqslant_{R}L$ such that $\dim_{R}(W)=k-r-1$, $\dim_{R}(I)=1$, it holds that $(S+W)\cap I\neq\{0\}$.
\end{proposition}

\begin{proof}
We begin by noting that for any $W,I\leqslant_{R}L$, it holds that $S\cap(I+W)\subseteq W$ if and only if $(S+W)\cap I\subseteq W$. Next, we establish the ``only if'' part. Let $W,I\leqslant_{R}L$ such that $\dim_{R}(W)=k-r-1$, $\dim_{R}(I)=1$. If $W\cap I\neq\{0\}$, then we have $I\subseteq W\subseteq S+W$ (since $0\in S$), which implies that $(S+W)\cap I=I\neq\{0\}$, as desired; if $W\cap I=\{0\}$, then we have $\dim_{R}(I+W)=k-r$, which implies that $\langle S\cap(I+W)\rangle_{R}=I+W$, and in particular, $S\cap(I+W)\nsubseteq W$, which further implies that $(S+W)\cap I\neq\{0\}$, as desired. Now we establish the ``if'' part. Let $V\leqslant_{R}L$ with $\dim_{R}(V)=k-r$. To prove that $\langle S\cap V\rangle_{R}=V$, it suffices to show that $S\cap V\nsubseteq W$ for an arbitrary $W\leqslant_{R}V$ with $\dim_{R}(W)=k-r-1$. Indeed, let $I\leqslant_{R}V$ such that $W+I=V$, $W\cap I=\{0\}$. Then, we have $(S+W)\cap I\neq \{0\}$, which, together with $0\in S$ and $W\cap I=\{0\}$, implies that $(S+W)\cap I\nsubseteq W$, which further implies that $S\cap V=S\cap(I+W)\nsubseteq W$, as desired.
\end{proof}

\subsection{Some preliminaries for $q$-binomial coefficient}
\setlength{\parindent}{2em}
For any $(n,r)\in\mathbb{N}^{2}$ and $q\in\mathbb{R}$, $q>1$, following \cite{26}, the $q$-binomial coefficient of $(n,r)$, denoted by $\bin_{q}(n,r)$, is defined as
\begin{equation}\bin_{q}(n,r)\triangleq\begin{cases}
0,&n+1\leqslant r,\\
\prod_{i=1}^{r}\frac{q^{i+n-r}-1}{q^{i}-1},&n\geqslant r;
\end{cases}
\end{equation}
moreover, define
\begin{equation}\delta_{q}(n,r)\triangleq\prod_{i=0}^{r-1}(q^{n}-q^{i}).\end{equation}
When $q$ is a prime power and $\mathbb{F}$ is a finite field with $|\mathbb{F}|=q$, it is well known that $\bin_{q}(n,r)$ is equal to the number of all the $r$-dimensional $\mathbb{F}$-subspaces of $\mathbb{F}^{n}$; moreover, for any $m,n\in\mathbb{Z}^{+}$ and $r\in\mathbb{N}$, $\delta_{q}(m,r)\bin_{q}(n,r)$ is equal to the number of all the matrices in $\Mat_{m,n}(\mathbb{F})$ with rank $r$.

\setlength{\parindent}{2em}
In the following proposition, we give two inequalities which will be used in Section 5.3 when we establish existence results for $r$-minimal codes.

\setlength{\parindent}{0em}
\begin{proposition}
{\bf{(1)}}\,\,For any $a\in\mathbb{R}$ with $a>1$ and $n\in\mathbb{N}$, it holds that
$$\prod_{i=1}^{n}(1-a^{-i})>1-a^{-1}-a^{-2}+a^{-5}>1-a^{-1}-a^{-2}.$$

{\bf{(2)}}\,\,For any $a\in\mathbb{R}$ with $a>\frac{1+\sqrt{5}}{2}$ and $(m,n,h)\in\mathbb{N}^{3}$ with $1\leqslant h\leqslant\min\{m,n\}$, it holds that
$$a^{-h(m+n-h)}\left(\sum_{i=0}^{h}\delta_{a}(m,i)\bin_{a}(n,i)\right)<\frac{1}{(a-1)(a^{2}-a-1)a^{m+n-2h-2}}+\frac{a^{2}}{a^{2}-a-1}.$$
\end{proposition}

\begin{proof}
{\bf{(1)}}\,\,Without loss of generality, suppose that $n\geqslant2$. From the Weierstrass inequality (see \cite{15}), we deduce that
$$\prod_{i=3}^{n}(1-a^{-i})\geqslant1-\left(\sum_{i=3}^{n}a^{-i}\right)>1-\frac{a^{-2}}{a-1},$$
which further implies that
$$\prod_{i=1}^{n}(1-a^{-i})>(1-a^{-1})(1-a^{-2})\left(1-\frac{a^{-2}}{a-1}\right)=1-a^{-1}-a^{-2}+a^{-5},$$
as desired.

{\bf{(2)}}\,\,Since $a>\frac{1+\sqrt{5}}{2}$, we have $1-a^{-1}-a^{-2}>0$. Consider $t\in\{0,1,\dots,h\}$. From (2.4), (2.5) and (1), we deduce that $0<\delta_{a}(m,t)\leqslant a^{mt}$ and
$$\bin_{a}(n,t)\leqslant\prod_{i=1}^{t}\frac{a^{i+n-t}}{a^{i}-1}=a^{t(n-t)}\left(\prod_{i=1}^{t}(1-a^{-i})\right)^{-1}<a^{t(n-t)}\cdot(1-a^{-1}-a^{-2})^{-1}.$$
It then follows that
\begin{equation}\sum_{i=0}^{h}\delta_{a}(m,i)\bin_{a}(n,i)<\sum_{i=0}^{h}a^{mi}\cdot a^{i(n-i)}\cdot(1-a^{-1}-a^{-2})^{-1}=\left(\sum_{i=0}^{h} a^{i(m+n-i)}\right)\cdot\frac{a^{2}}{a^{2}-a-1}.\end{equation}
Noticing that $i(m+n-i)<j(m+n-j)$ for all $i,j\in\mathbb{N}$ with $i<j\leqslant\min\{m,n\}$, we have
$$\sum_{i=0}^{h-1} a^{i(m+n-i)}\leqslant\sum_{s=0}^{(h-1)(m+n-h+1)} a^{s}<\frac{a^{(h-1)(m+n-h+1)+1}}{a-1},$$
which, together with (2.6), implies that
$$\sum_{i=0}^{h}\delta_{a}(m,i)\bin_{a}(n,i)<\left(\frac{a^{(h-1)(m+n-h+1)+1}}{a-1}+a^{h(m+n-h)}\right)\cdot\frac{a^{2}}{a^{2}-a-1},$$
which further establishes (2) via some straightforward computation which we omit.
\end{proof}

\section{General $r$-minimal codes supported on lattices}
\setlength{\parindent}{2em}
Throughout this section, let $(Y,\preccurlyeq)$ be a lattice with the least element $0_{Y}$. Necessarily, $(Y,\preccurlyeq)$ is a poset, and for any $u,v\in Y$, $\{u,v\}$ has a supremum in $(Y,\preccurlyeq)$, which we denote by $u\curlyvee v$ (see [20, Definitions 2.1 and 2.4]). A \textit{chain} in $(Y,\preccurlyeq)$ is a subset $I\subseteq Y$ such that for any $u,v\in I$, either $u\preccurlyeq v$ or $v\preccurlyeq u$ holds. We suppose that $(Y,\preccurlyeq)$ \textit{has length $n\in\mathbb{N}$}, i.e., all the chains in $(Y,\preccurlyeq)$ are finite and has cardinality less than or equal to $n+1$, and the largest cardinality of a chain in $(Y,\preccurlyeq)$ is equal to $n+1$ (see [20, Definition 2.37]). Let $\rk:Y\longrightarrow\mathbb{N}$ be the length function of $(Y,\preccurlyeq)$, i.e.,
$$\forall~y\in Y:\rk(y)=\max\{|I|-1\mid \text{$I$ is a chain in $(Y,\preccurlyeq)$ containing $y$ as its greatest element}\}.$$
Let $R$ be a division ring, and let $X$ be an $n$-dimensional left $R$-module. Any $R$-submodule of $X$ is referred to as a \textit{code}. Let $\sigma:X\longrightarrow Y$ satisfy the following three conditions:
\begin{equation}\sigma(0)=0_{Y},\end{equation}
\begin{equation}\forall~\alpha\in X,\forall~\beta\in X:\sigma(\alpha+\beta)\preccurlyeq\sigma(\alpha)\curlyvee\sigma(\beta),\end{equation}
\begin{equation}\forall~c\in R,\forall~\beta\in X:\sigma(c\beta)\preccurlyeq\sigma(\beta).\end{equation}
Define $\mu:Y\longrightarrow2^{X}$ as
\begin{equation}\forall~y\in Y:\mu(y)=\{\alpha\in X\mid\sigma(\alpha)\preccurlyeq y\}.\end{equation}
It follows from (3.1)--(3.3) that $\mu(y)\leqslant_{R}X$ for all $y\in Y$. We further assume that
\begin{equation}\forall~y\in Y:\dim_{R}(\mu(y))=\rk(y).\end{equation}
Since $(Y,\preccurlyeq)$ has finite length $n$, any subset of $Y$ has a supremum in $(Y,\preccurlyeq)$. Now, we define $\lambda:2^{X}\longrightarrow Y$ as
\begin{equation}\forall~A\subseteq X:\lambda(A)=\text{the supremum of $\{\sigma(\alpha)\mid\alpha\in A\}$ in $(Y,\preccurlyeq)$}.\end{equation}

\setlength{\parindent}{2em}
Now we are ready to define $r$-minimal codes with respect to the ``support'' $\sigma$.

\setlength{\parindent}{0em}
\begin{definition}
{\bf{(1)}}\,\,Let $C\leqslant_{R}X$. A submodule $D\leqslant_{R}C$ is said to be $\sigma$-minimal in $C$ if for any $B\leqslant_{R}C$ such that $\dim_{R}(B)=\dim_{R}(D)$, $\lambda(B)\preccurlyeq \lambda(D)$, it holds that $\lambda(B)=\lambda(D)$. Moreover, for $r\in\mathbb{N}$, $C$ is said to be $(\sigma,r)$-minimal if every $r$-dimensional $R$-submodule of $C$ is $\sigma$-minimal in $C$.

{\bf{(2)}}\,\,Let $C\leqslant_{R}X$. A submodule $D\leqslant_{R}C$ is said to be $\sigma$-maximal in $C$ if for any $B\leqslant_{R}C$ such that $\dim_{R}(B)=\dim_{R}(D)$, $\lambda(D)\preccurlyeq \lambda(B)$, it holds that $B=D$.
\end{definition}

\setlength{\parindent}{2em}
Definition 3.1 is general in the sense that with $\sigma$ appropriately set, it includes $r$-minimal codes (in particular, minimal codes) with respect to various metrics as special cases; moreover, the term ``$\lambda(B)=\lambda(D)$'' in (1) of Definition 3.1 can be replaced by ``$B=D$'' (see Section 3.1 and Theorem 3.2 for more details). We also note that $\sigma$-maximal subcodes is inspired by the notion of maximal codewords proposed in [2, Remark 1.4].

\subsection{Examples}
\setlength{\parindent}{2em}
In this subsection, by setting $\sigma$ appropriately, we consider $r$-minimal codes with respect to various settings and discuss some of their basic properties. We begin with rank metric codes.

\begin{example}
Let $\mathbb{E}/\mathbb{F}$ be a finite dimensional field extension, and fix $n\in\mathbb{Z}^{+}$. Set $R=\mathbb{E}$, $X=\mathbb{E}^{n}$, $Y=\{U\mid U\leqslant_{\mathbb{F}}\mathbb{F}^{n}\}$; moreover, set $\sigma(\alpha)=\rsupp(\alpha)$, $\lambda(A)=\chi(A)$ as in Section 2.1, respectively, and set $\mu(U)=\langle U\rangle_{\mathbb{E}}$ for all $U\in Y$. Then, $(Y,\subseteq)$ is a lattice of length $n$, $\rk(U)=\dim_{\mathbb{F}}(U)$ for all $U\in Y$, and $\sigma$, $\mu$, $\lambda$ satisfy (3.1)--(3.6). Now for $C\leqslant_{\mathbb{E}}\mathbb{E}^{n}$ and $D\leqslant_{\mathbb{E}}C$, one can check that $D$ is $\sigma$-minimal in $C$ if and only if $D$ is rank minimal in $C$; moreover, for any $r\in\mathbb{N}$, $C$ is $(\sigma,r)$-minimal if and only if $C$ is $r$-minimal with respect to rank metric. Furthermore, in Section 4, we will characterize rank minimal subcodes and $r$-minimal rank metric codes in terms of generator matrices and cutting $r$-blocking sets (see Theorem 4.3 for more details).
\end{example}

\setlength{\parindent}{2em}
Next, we consider $r$-minimal Hamming metric codes.

\begin{example}
Suppose that $R$ is a finite field, and fix $n\in\mathbb{Z}^{+}$. Set $X=R^{n}$, $Y=2^{\{1,\dots,n\}}$; moreover, set $\sigma$, $\mu$ and $\lambda$ as $\sigma(\alpha)=\{i\in\{1,\dots,n\}\mid\alpha_i\neq0\}$, $\mu(I)=\{\alpha\in X\mid\sigma(\alpha)\subseteq I\}$ and $\lambda(A)=\bigcup_{\alpha\in A}\sigma(\alpha)$, respectively. Then, $(Y,\subseteq)$ is a lattice of length $n$, $\rk(I)=|I|$ for all $I\in Y$, and $\sigma$, $\mu$, $\lambda$ satisfy (3.1)--(3.6).

Now consider $C\leqslant_{R}X$ with $\dim_{R}(C)=k\geqslant1$. For any $D\leqslant_{R}C$, we say that $D$ is Hamming minimal (\textit{resp.}, Hamming maximal) in $C$ if $D$ is $\sigma$-minimal (\textit{resp.}, $\sigma$-maximal) in $C$. Moreover, for any $r\in\mathbb{N}$, we say that $C$ is $r$-minimal with respect to Hamming metric if $C$ is $(\sigma,r)$-minimal. When $r=1$, these definitions recover the definitions of minimal codewords, maximal codewords and minimal codes with respect to Hamming metric (see, e.g., [2, Definition 1.3, Remark 1.4]). The $r$-minimality of $C$ can be characterized in terms of its generator matrix. Indeed, let $G\in\Mat_{k,n}(R)$ such that $C=\langle\row(G)\rangle_{R}$, and let $S=\col(G)$. Then, for $B\leqslant_{R}R^{k}$ and $D\triangleq\{\gamma G\mid\gamma\in B\}$, $D$ is Hamming minimal in $C$ if and only if $\langle{B^{\ddagger}}\cap S\rangle_{R}=B^{\ddagger}$; moreover, for $r\in\{0,1,\dots,k\}$, $C$ is $r$-minimal if and only if $S$ is a cutting $r$-blocking set of $R^{[k]}$. These two facts, whose proofs are omitted, generalize the well known equivalence between minimal Hamming metric codes and cutting blocking sets (see [1, Theorem 3.4], [46, Theorem 3.2]).
\end{example}

\setlength{\parindent}{2em}
In the following example, we consider $r$-minimal codes with respect to sum-rank metric (see \cite{13,34,38,44}), a metric that includes both Hamming metric and rank metric as special cases.

\begin{example}
Let $\mathbb{E}/\mathbb{F}$ be a finite dimensional field extension, $t\in\mathbb{Z}^{+}$, and $(w_1,\dots,w_t)\in(\mathbb{Z}^{+})^{t}$. For $i\in\{1,\dots,t\}$, let $Q_i=\{U\mid U\leqslant_{\mathbb{F}}\mathbb{F}^{w_i}\}$, and define $\rsupp(\alpha)$ and $\chi(A)$ as in Section 2.1 for all $\alpha\in\mathbb{E}^{w_i}$ and $A\subseteq\mathbb{E}^{w_i}$. Set $R=\mathbb{E}$, $X=\prod_{i=1}^{t}\mathbb{E}^{w_i}$, $Y=\prod_{i=1}^{t}Q_i$, $n=\sum_{i=1}^{t}w_{i}$. Let $\preccurlyeq$ denote the partial order on $Y$ such that $(L_1,\dots,L_t)\preccurlyeq(M_1,\dots,M_t)$ if and only if $L_i\subseteq M_{i}$ for all $i\in\{1,\dots,t\}$. Moreover, set $\sigma$, $\mu$ and $\lambda$ as $\sigma(\theta)=(\rsupp(\theta_1),\dots,\rsupp(\theta_t))$, $\mu(U_1,\dots,U_{t})=\prod_{i=1}^{t}\langle U_{i}\rangle_{\mathbb{E}}$ and $\lambda(B)=(\chi(\{\theta_1\mid\theta\in B\}),\dots,\chi(\{\theta_t\mid\theta\in B\}))$, respectively. Then, $(Y,\preccurlyeq)$ is a lattice of length $n$, $\rk(U_1,\dots,U_{t})=\sum_{i=1}^{t}\dim_{\mathbb{F}}(U_{i})$ for all $(U_1,\dots,U_{t})\in Y$, and $\sigma$, $\mu$, $\lambda$ satisfy (3.1)--(3.6).

Now consider $C\leqslant_{\mathbb{E}}X$ with $\dim_{\mathbb{E}}(C)=k$. $C$ is referred to as a $[(w_1,\dots,w_t),k]$ sum-rank metric code. For any $D\leqslant_{\mathbb{E}}C$, we say that $D$ is sum-rank minimal in $C$ if $D$ is $\sigma$-minimal in $C$. Moreover, for any $r\in\mathbb{N}$, we say that $C$ is $r$-minimal with respect to sum-rank metric if $C$ is $(\sigma,r)$-minimal. With $r$ set to be $1$, these definitions recover the definitions of minimal codewords and minimal codes with respect to sum-rank metric (see [13, Definition 2.1], [44, Definition 9.21]). We will show in Section 4 that similar to rank metric and Hamming metric, the $r$-minimality of $C$ can be characterized in terms of matrices and cutting $r$-blocking sets (see Remark 4.3 for more details).
\end{example}

\setlength{\parindent}{2em}
As in Examples 3.1 and 3.3, sum-rank metric may be roughly regarded as the Cartesian product of rank metric. Such an observation leads to the following more general construction.

\begin{example}
Let $\Lambda$ be a nonempty finite set. For any $i\in\Lambda$, let $(Q_i,\lessdot_{i})$ be a lattice of length $w_i\in\mathbb{N}$ with the length function $\rho_i:Q_i\longrightarrow\mathbb{N}$, $P_{i}$ be an $w_i$-dimensional left $R$-module, and let $\xi_{i}:P_{i}\longrightarrow Q_{i}$ satisfy (3.1)--(3.3); moreover, define $\eta_{i}: Q_{i}\longrightarrow2^{P_{i}}$ and $\tau_{i}: 2^{P_{i}}\longrightarrow Q_{i}$ as in (3.4) and (3.6), with $X$, $Y$ and $\sigma$ replaced by $P_i$, $Q_i$ and $\xi_i$, respectively, and suppose that $\eta_{i}$ and $\rho_{i}$ satisfy (3.5). Set $Y=\prod_{i\in\Lambda}Q_i$, $n=\sum_{i\in\Lambda}w_{i}$, $X=\prod_{i\in\Lambda}P_i$, and let $\preccurlyeq$ denote the partial order on $Y$ such that $g\preccurlyeq h$ if and only if $g_i\lessdot_{i} h_i$ for all $i\in\Lambda$. Moreover, set $\sigma$, $\mu$ and $\lambda$ as $\sigma(\alpha)=(\xi_{i}(\alpha_i)\mid i\in\Lambda)$, $\mu(y)=\prod_{i\in\Lambda}\eta_{i}(y_i)$ and $\lambda(B)=(\tau_{i}(\{\beta_{i}\mid \beta\in B\})\mid i\in\Lambda)$, respectively. Then, $(Y,\preccurlyeq)$ is a lattice of length $n$, $\rk(g)=\sum_{i\in\Lambda}\rho_i(g_i)$ for all $g\in Y$, and $\sigma$, $\mu$, $\lambda$ satisfy (3.1)--(3.6).

Now let $C_{i}\leqslant_{R}P_{i}$ for all $i\in\Lambda$, and let $M\triangleq\prod_{i\in\Lambda}C_{i}$. Then, for $B\leqslant_{R}M$, we infer that $B$ is $\sigma$-minimal in $M$ if and only if $B=\prod_{i\in\Lambda}D_{i}$ for some $(D_{i}\mid i\in\Lambda)$, where $D_{i}\leqslant_{R}C_{i}$ and $D_{i}$ is $\xi_{i}$-minimal in $C_{i}$ for all $i\in\Lambda$. Moreover, for $r\in\{1,\dots,\dim_{R}(M)-1\}$, $M$ is $(\sigma,r)$-minimal if and only if there exists $l\in\Lambda$ such that $C_l$ is $(\xi_{l},r)$-minimal and $C_{i}=\{0\}$ for all $i\in\Lambda-\{l\}$. We omit the details of the proofs of these two facts.
\end{example}

\setlength{\parindent}{2em}
Next, we consider the notion of outer minimal codes, which has recently been introduced in \cite{4}.

\begin{example}
Suppose that $R$ is a field and $\mathbb{P}$ is an $m$-dimensional field extension of $R$. Let $t\in\mathbb{Z}^{+}$. Set $X=\mathbb{P}^{t}$, $Y=\{U\mid U\leqslant_{R}\mathbb{P}\}^{t}$ and $n=mt$. Moreover, set $\sigma$, $\mu$ and $\lambda$ as $\sigma(\alpha)=(R\alpha_1,\dots,R\alpha_t)$, $\mu(L_1,\dots,L_t)=\prod_{i=1}^{t}L_i$ and $\lambda(B)=(\langle\{\beta_{1}\mid \beta\in B\}\rangle_{R},\dots,\langle\{\beta_{t}\mid \beta\in B\}\rangle_{R})$, respectively. Then, $(Y,\preccurlyeq)$ is a lattice of length $n$, $\rk(L_1,\dots,L_t)=\sum_{i\in\Lambda}\dim_{R}(L_i)$ for all $(L_1,\dots,L_t)\in Y$, and $\sigma$, $\mu$, $\lambda$ satisfy (3.1)--(3.6).

Let $C\leqslant_{R}X$. For any $D\leqslant_{R}C$, we say that $D$ is outer minimal in $C$ if $D$ is $\sigma$-minimal in $C$; moreover, for any $r\in\mathbb{N}$, we say that $C$ is outer $r$-minimal if $C$ is $(\sigma,r)$-minimal. When $r=1$, these definitions recover the definitions of outer minimal codewords and outer minimal codes which were proposed for $\mathbb{P}$-subspaces of $X$ (see [4, Definitions 21 and 23]). Now suppose that $C$ is a $\mathbb{P}$-subspace of $X$ with $\dim_{\mathbb{P}}(C)=k\geqslant1$. Let $G\in\Mat_{k,t}(\mathbb{P})$ such that $C=\langle\row(G)\rangle_{\mathbb{P}}$, and let $S=\{b\beta\mid b\in\mathbb{P},\beta\in\col(G)\}$. Then, for $r\in\{0,1,\dots,k\}$, $C$ is outer $r$-minimal if and only if $S$ is a cutting $r$-blocking set of the $R$-vector space $\mathbb{P}^{[k]}$. This result, whose proof is omitted, recovers [4, Proposition 24] and [4, Theorem 20 (d)] when $r=1$.

Finally, if we identify $\mathbb{P}$ with $R^{m}$ as $R$-vector spaces, then this example can be alternatively derived from Example 3.3 by letting $\mathbb{E}=\mathbb{F}=R$ and $w_{i}=m$ for all $i\in\{1,\dots,t\}$. Hence outer $r$-minimal codes can be regarded as a special case of $r$-minimal sum-rank metric codes.
\end{example}

\setlength{\parindent}{2em}
We end this subsection by introducing $r$-minimal codes with respect to poset metric, a metric first introduced in \cite{14} that naturally generalizes Hamming metric (also see \cite{28,37}). The following example seems to be new in the theory of minimal codes.

\begin{example}
Suppose that $R$ is a finite field. Let $n\in\mathbb{Z}^{+}$, and let $\mathbf{P}=(\{1,\dots,n\},\preccurlyeq_{\mathbf{P}})$ be a poset. A subset $I\subseteq\{1,\dots,n\}$ is referred to as an ideal of $\mathbf{P}$ if for any $v\in I$ and $u\in\{1,\dots,n\}$ with $u\preccurlyeq_{\mathbf{P}}v$, it holds that $u\in I$. Set $X=R^{n}$, $Y=\{I\mid\text{$I$ is an ideal of $\mathbf{P}$}\}$; moreover, set $\sigma$, $\mu$ and $\lambda$ as $\sigma(\alpha)=\{i\mid\exists~j\in\{1,\dots,n\}~s.t.~\alpha_j\neq0,i\preccurlyeq_{\mathbf{P}}j\}$,
$\mu(I)=\{\alpha\in X\mid\sigma(\alpha)\subseteq I\}$ and $\lambda(A)=\bigcup_{\alpha\in A}\sigma(\alpha)$, respectively. It follows from [28, Proposition 1.1] that $(Y,\subseteq)$ is a lattice of length $n$, $\rk(I)=|I|$ for all $I\in Y$, and $\sigma$, $\mu$, $\lambda$ satisfy (3.1)--(3.6).

Let $C\leqslant_{R}X$ with $\dim_{R}(C)=k\geqslant1$. For any $D\leqslant_{R}C$, we say that $D$ is $\mathbf{P}$-minimal in $C$ if $D$ is $\sigma$-minimal in $C$. Moreover, for any $r\in\mathbb{N}$, we say that $C$ is $r$-minimal with respect to $\mathbf{P}$ if $C$ is $(\sigma,r)$-minimal. Now let $G\in\Mat_{k,n}(R)$ such that $C=\langle\row(G)\rangle_{R}$, and let $\beta(j)$ denote the $j$-th column of $G$. Then, for $B\leqslant_{R}R^{k}$ and $D\triangleq\{\gamma G\mid\gamma\in B\}$, $D$ is $\mathbf{P}$-minimal in $C$ if and only if
\begin{equation}\langle\{\beta(i)\mid i\in\lambda(C)~s.t.~(\forall~j\in\lambda(C):i\preccurlyeq_{\mathbf{P}}j\Longrightarrow\beta(j)\in B^{\ddagger})\}\rangle_{R}=B^{\ddagger};\end{equation}
moreover, for $r\in\{0,1,\dots,k\}$, $C$ is $r$-minimal with respect to $\mathbf{P}$ if and only if for any $L\leqslant_{R}R^{[k]}$ with $\dim_{R}(L)=k-r$, it holds that
\begin{equation}\langle\{\beta(i)\mid i\in\lambda(C)~s.t.~(\forall~j\in\lambda(C):i\preccurlyeq_{\mathbf{P}}j\Longrightarrow\beta(j)\in L)\}\rangle_{R}=L.\end{equation}
We omit the proofs of (3.7) and (3.8). Instead, we remark that if $\mathbf{P}$ is an anti-chain, then (3.8) holds true if and only if $\col(G)$ is a cutting $r$-blocking set of $R^{[k]}$, which further recovers the equivalence between $r$-minimal Hamming metric codes and cutting $r$-blocking sets established in Example 3.2.
\end{example}

\subsection{$\sigma$-Minimal subcodes and $(\sigma,r)$-minimal codes}

\setlength{\parindent}{2em}
In this subsection, we characterize $\sigma$-minimal subcodes and $(\sigma,r)$-minimal codes, and derive some of their basic properties. We begin with the following lemma, whose proof is straightforward and hence omitted.

\setlength{\parindent}{0em}
\begin{lemma}
{\bf{(1)}}\,\,For any $B\subseteq X$ and $A\subseteq B$, it holds that $\lambda(A)\preccurlyeq\lambda(B)$.

{\bf{(2)}}\,\,For any $y\in Y$, $A\subseteq X$, it holds that $\lambda(A)\preccurlyeq y\Longleftrightarrow A\subseteq\mu(y)$.

{\bf{(3)}}\,\,For any $A\subseteq X$, it holds that $\mu(\lambda(A))\leqslant_{R}X$, $\dim_{R}(\mu(\lambda(A)))=\rk(\lambda(A))$ and $A\subseteq\mu(\lambda(A))$.

{\bf{(4)}}\,\,For any $U\leqslant_{R}X$, $V\leqslant_{R}X$, it holds that $\lambda(U+V)=\lambda(U)\curlyvee\lambda(V)$.

{\bf{(5)}}\,\,For $C\leqslant_{R}X$, $C$ is both $(\sigma,0)$-minimal and $(\sigma,r)$-minimal for all $r\in\mathbb{N}$ with $r\geqslant\dim_{R}(C)$; moreover, for $W\leqslant_{R}C$ and $D\leqslant_{R}W$, if $D$ is $\sigma$-minimal in $C$, then $D$ is $\sigma$-minimal in $W$.
\end{lemma}

\setlength{\parindent}{2em}
The following theorem is inspired by [47, Theorem 1], and is crucial for our discussion.

\setlength{\parindent}{0em}
\begin{theorem}
Let $V\leqslant_{R}X$, $V\neq\{0\}$. Then, there exists $U\leqslant_{R}V$ such that $\dim_{R}(U)=\dim_{R}(V)-1$, $\lambda(U)\neq\lambda(V)$.
\end{theorem}

\begin{proof}
By (3) of Lemma 3.1, we have $\rk(\lambda(V))\geqslant1$. Hence we can choose $w\in Y$ such that $w\preccurlyeq\lambda(V)$, $\rk(w)=\rk(\lambda(V))-1$. Let $U\triangleq V\cap\mu(w)\leqslant_{R}V$. By (2) of Lemma 3.1, we have $\lambda(U)\preccurlyeq w$, which implies that $\lambda(U)\neq\lambda(V)$, as desired. Next, we show that $\dim_{R}(U)=\dim_{R}(V)-1$. Indeed, from $w\preccurlyeq\lambda(V)$ and (3.5), we deduce that $\mu(w)\leqslant_{R}\mu(\lambda(V))$ and $\dim_{R}(\mu(w))=\dim_{R}(\mu(\lambda(V)))-1$; moreover, from (2), (3) of Lemma 3.1 and $\lambda(V)\not\preccurlyeq w$, we deduce that $V\leqslant_{R}\mu(\lambda(V))$ and $V\nsubseteq\mu(w)$. It then follows that $\dim_{R}(V)-\dim_{R}(V\cap\mu(w))=1$, as desired.
\end{proof}

\begin{remark}
With respect to $\rk$ and $\sigma$, we can define generalized weights for a code. More precisely, let $C\leqslant_{R}X$ with $\dim_{R}(C)=k$. For any $r\in\{0,1,\dots,k\}$, define
$$\phi_{r}(C)=\min\{\rk(\lambda(D))\mid D\leqslant_{R}C,\dim_{R}(D)=r\}.$$
It follows from Theorem 3.1 that $\phi_{r}(C)<\phi_{r+1}(C)$ for all $r\in\{0,1,\dots,k-1\}$, which, with $\sigma$ appropriately set, further recovers [23, Theorem I.2], [47, Theorem 1], [37, Proposition 1] and a special case of [34, Lemma 4] which were established for generalized rank weights, generalized Hamming weights, generalized poset weights and generalized sum-rank weights, respectively.
\end{remark}

\setlength{\parindent}{2em}
Now we prove the two main results of this subsection. First, we give seven necessary and sufficient conditions for a subcode to be $\sigma$-minimal.

\setlength{\parindent}{0em}
\begin{theorem}
Let $C\leqslant_{R}X$, and let $D\leqslant_{R}C$. Then, the following eight statements are equivalent to each other:

{\bf{(1)}}\,\,$D$ is $\sigma$-minimal in $C$;

{\bf{(2)}}\,\,For any $W\leqslant_{R}C$ such that $D\subseteq W$, $\dim_{R}(W)=\dim_{R}(D)+1$, $D$ is $\sigma$-minimal in $W$;

{\bf{(3)}}\,\,For any $W\leqslant_{R}C$ such that $D\subseteq W$, $\dim_{R}(W)=\dim_{R}(D)+1$, it holds that $\lambda(D)\neq\lambda(W)$;

{\bf{(4)}}\,\,$C\cap\mu(\lambda(D))=D$;

{\bf{(5)}}\,\,$\dim_{R}(C)-\dim_{R}(D)=\dim_{R}(\mu(\lambda(D))+C)-\rk(\lambda(D))$;

{\bf{(6)}}\,\,For any $A\subseteq C$ with $\lambda(A)\preccurlyeq\lambda(D)$, it holds that $A\subseteq D$;

{\bf{(7)}}\,\,For any $B\leqslant_{R}C$ such that $\dim_{R}(B)=\dim_{R}(D)$, $\lambda(B)\preccurlyeq \lambda(D)$, it holds that $B=D$;

{\bf{(8)}}\,\,For any $z\in Y$ such that $z\preccurlyeq\lambda(D)$, $\dim_{R}(C\cap\mu(z))\geqslant\dim_{R}(D)$, it holds that $z=\lambda(D)$.
\end{theorem}

\begin{proof}
$(1)\Longrightarrow(2)$\,\,This immediately follows from (5) of Lemma 3.1.

$(2)\Longrightarrow(3)$\,\,Suppose by way of contradiction that $\lambda(D)=\lambda(W)$ for some $W\leqslant_{R}C$ such that $D\subseteq W$, $\dim_{R}(W)=\dim_{R}(D)+1$. By Theorem 3.1, we can choose $B\leqslant_{R}W$ such that $\dim_{R}(B)=\dim_{R}(W)-1$, $\lambda(B)\neq\lambda(W)$. It follows that $\dim_{R}(B)=\dim_{R}(D)$, $\lambda(B)\neq\lambda(D)$, $\lambda(B)\preccurlyeq\lambda(W)=\lambda(D)$, a contradiction to the fact that $D$ is $\sigma$-minimal in $W$, as desired.

$(3)\Longrightarrow(4)$\,\,Suppose by way of contradiction that $D\neq C\cap\mu(\lambda(D))$. Then, from $D\subseteq C\cap\mu(\lambda(D))$, we can choose $W\leqslant_{R}X$ such that $D\subseteq W\subseteq C\cap\mu(\lambda(D))$, $\dim_{R}(W)=\dim_{R}(D)+1$. By (1), (2) of Lemma 3.1, we have $\lambda(D)=\lambda(W)$, a contradiction to (3), as desired.

$(4)\Longleftrightarrow(5)$\,\,From $D\subseteq C\cap\mu(\lambda(D))$ and $\dim_{R}(\mu(\lambda(D)))=\rk(\lambda(D))$, we deduce that
\begin{eqnarray*}
\begin{split}
C\cap\mu(\lambda(D))=D&\Longleftrightarrow\dim_{R}(C\cap\mu(\lambda(D)))=\dim_{R}(D)\\
&\Longleftrightarrow\dim_{R}(C)+\rk(\lambda(D))-\dim_{R}(\mu(\lambda(D))+C)=\dim_{R}(D),
\end{split}
\end{eqnarray*}
which immediately establishes $(4)\Longleftrightarrow(5)$, as desired.

$(4)\Longrightarrow(6)$\,\,Let $A\subseteq C$ with $\lambda(A)\preccurlyeq\lambda(D)$. By (2) of Lemma 3.1, we have $A\subseteq\mu(\lambda(D))$, which, together with (4), implies that $A\subseteq C\cap\mu(\lambda(D))=D$, as desired.

$(6)\Longrightarrow(7)$\,\,Let $B\leqslant_{R}C$ such that $\dim_{R}(B)=\dim_{R}(D)$, $\lambda(B)\preccurlyeq \lambda(D)$. By (6), we have $B\subseteq D$, which, together with $\dim_{R}(B)=\dim_{R}(D)$, implies that $B=D$, as desired.

$(7)\Longrightarrow(1)$\,\,This immediately follows from Definition 3.1.

$(6)\Longrightarrow(8)$\,\,Let $z\in Y$ such that $z\preccurlyeq\lambda(D)$, $\dim_{R}(C\cap\mu(z))\geqslant\dim_{R}(D)$, and let $B=C\cap\mu(z)$. By (2) of Lemma 3.1, we have $\lambda(B)\preccurlyeq z\preccurlyeq\lambda(D)$. It then follows from (6) that $B\subseteq D$, which, together with $\dim_{R}(B)\geqslant\dim_{R}(D)$, implies that $B=D$, which further implies that $z=\lambda(D)$, as desired.

$(8)\Longrightarrow(1)$\,\,By way of contradiction, suppose that $D$ is not $\sigma$-minimal in $C$. Then, we can choose $B\leqslant_{R}C$ such that $\dim_{R}(B)=\dim_{R}(D)$, $\lambda(B)\preccurlyeq\lambda(D)$, $\lambda(B)\neq\lambda(D)$. From $B\subseteq C\cap\mu(\lambda(B))$, we deduce that $\dim_{R}(C\cap\mu(\lambda(B)))\geqslant\dim_{R}(B)=\dim_{R}(D)$, a contradiction to (8), as desired.
\end{proof}

\setlength{\parindent}{2em}
Next, we give three necessary and sufficient conditions for code to be $(\sigma,r)$-minimal. The following theorem is a direct corollary of Theorem 3.2 and (2) of Definition 3.1.

\setlength{\parindent}{0em}
\begin{theorem}
For $C\leqslant_{R}X$ and $r\in\mathbb{N}$, the following four statements are equivalent to each other:

{\bf{(1)}}\,\,$C$ is $(\sigma,r)$-minimal;

{\bf{(2)}}\,\,For any $D,B\leqslant_{R}C$ such that $\dim_{R}(D)=\dim_{R}(B)=r$, $\lambda(D)\preccurlyeq \lambda(B)$, it holds that $D=B$;

{\bf{(3)}}\,\,Every $r$-dimensional $R$-submodule of $C$ is $\sigma$-maximal in $C$;

{\bf{(4)}}\,\,For any $W\leqslant_{R}C$ with $\dim_{R}(W)=r+1$, and for any $D\leqslant_{R}W$ with $\dim_{R}(D)=r$, it holds that $\lambda(D)\neq\lambda(W)$.
\end{theorem}

\setlength{\parindent}{2em}
Regarding $\rk(\lambda(D))$ as the ``$\sigma$-support weight'' of $D\subseteq X$, we now show that a code whose $r$-dimensional subcodes have constant $\sigma$-support weight is necessarily $(\sigma,r)$-minimal, as detailed in the following simple yet general result.

\setlength{\parindent}{0em}
\begin{proposition}
{\bf{(1)}}\,\,Let $C\leqslant_{R}X$, and let $D\leqslant_{R}C$. Suppose that for any $B\leqslant_{R}C$ with $\dim_{R}(B)=\dim_{R}(D)$, it holds that $\rk(\lambda(D))\leqslant\rk(\lambda(B))$. Then, $D$ is $\sigma$-minimal in $C$.

{\bf{(2)}}\,\,Let $C\leqslant_{R}X$, and let $r\in\mathbb{N}$. Suppose that for any $D,B\leqslant_{R}C$ such that $\dim_{R}(D)=\dim_{R}(B)=r$, it holds that $\rk(\lambda(D))=\rk(\lambda(B))$. Then, $C$ is $(\sigma,r)$-minimal.
\end{proposition}

\begin{proof}
{\bf{(1)}}\,\,Let $B\leqslant_{R}C$ such that $\dim_{R}(B)=\dim_{R}(D)$, $\lambda(B)\preccurlyeq \lambda(D)$. It then follows from $\rk(\lambda(D))\leqslant\rk(\lambda(B))$ and $\lambda(B)\preccurlyeq \lambda(D)$ that $\lambda(B)=\lambda(D)$, as desired.

{\bf{(2)}}\,\,This immediately follows from (1) and (1) of Definition 3.1.
\end{proof}

\begin{remark}
With $r$ set to be $1$, (2) of Proposition 3.1 recovers [3, Theorem 6.1] and [44, Proposition 9.25] which have been established for rank metric codes and sum-rank metric codes, respectively. We note that a complete classification of constant weight rank metric codes over finite fields has been given in \cite{38,3}, which we will further generalize to codes over possibly infinite fields and arbitrary $r$ in Section 4.3; on the other hand, constant weight sum-rank metric codes form a very large family (see \cite{13,38} for more details). Moreover, it has been proven in \cite{15,32} that constant weight Hamming metric linear codes are exactly repetitions of the dual of Hamming codes; however, an application of (2) of Proposition 3.1 to poset metric seems to yield a new result.
\end{remark}

\setlength{\parindent}{2em}
Now we derive some corollaries of the results established in this subsection so far. We begin with the following Singleton-type bound based on $\sigma$-minimal subcodes, which includes several previously known Singleton bounds as special cases.

\setlength{\parindent}{0em}
\begin{corollary}
{\bf{(1)}}\,\,Let $C\leqslant_{R}X$, $D\leqslant_{R}C$ such that $D$ is $\sigma$-minimal in $C$. Then, it holds that
\begin{equation}\dim_{R}(C)-\dim_{R}(D)\leqslant\rk(\lambda(C))-\rk(\lambda(D)).\end{equation}
Moreover, equality holds in (3.9) if and only if $\mu(\lambda(D))+C=\mu(\lambda(C))$.

{\bf{(2)}}\,\,Let $C\leqslant_{R}X$ with $\dim_{R}(C)=k$, and let $r\in\{0,1,\dots,k\}$ such that $C$ is $(\sigma,r)$-minimal. Then, for any $D\leqslant_{R}C$ with $\dim_{R}(D)=r$, it holds that $\rk(\lambda(D))\leqslant\rk(\lambda(C))-k+r$.
\end{corollary}

\begin{proof}
Part (1) immediately follows from ``$(1)\Longleftrightarrow(5)$'' of Theorem 3.2, together with the facts that $\mu(\lambda(D))+C\leqslant_{R}\mu(\lambda(C))$ and $\dim_{R}(\mu(\lambda(C)))=\rk(\lambda(C))$. Moreover, (2) immediately follows from (1) and Definition 3.1, as desired.
\end{proof}

\setlength{\parindent}{0em}
\begin{remark}
With $\sigma$ appropriately set, (1) of Corollary 3.1 recovers the generalized Singleton bounds in [23, Corollary II.4], [47, Corollary 1], [37, Proposition 2] and a special case of [34, Proposition 5] which were established for generalized rank weights, generalized Hamming weights, generalized poset weights and generalized sum-rank weights, respectively, and when $r=1$, it also recovers [2, Proposition 1.5] which was established for minimal codewords in Hamming metric. Moreover, when $r=1$, (2) of Corollary 3.1 recovers [3, Corollary 5.9] and [13, Theorem 2.4] which were established for minimal rank metric codes and minimal sum-rank metric codes, respectively.
\end{remark}

\setlength{\parindent}{2em}
Next, an application of Theorem 3.3 immediately yields the following result.

\setlength{\parindent}{0em}
\begin{corollary}
Let $(k,r)\in\mathbb{N}^{2}$ with $r+1\leqslant k\leqslant n$, and fix $t\in\{r+1,\dots,k\}$. Let $C\leqslant_{R}X$ with $\dim_{R}(C)=k$. Then, $C$ is $(\sigma,r)$-minimal if and only if every $t$-dimensional $R$-submodule of $C$ is $(\sigma,r)$-minimal.
\end{corollary}

\begin{proof}
It suffices to prove the ``if'' part. Let $W\leqslant_{R}C$ with $\dim_{R}(W)=r+1$, and let $D\leqslant_{R}W$ with $\dim_{R}(D)=r$. Then, we can choose $L\leqslant_{R}C$ such that $\dim_{R}(L)=t$, $W\subseteq L$. Since $L$ is $(\sigma,r)$-minimal, Theorem 3.3 implies that $\lambda(D)\neq\lambda(W)$, which further implies the desired result.
\end{proof}

\setlength{\parindent}{2em}
Finally, we show that $(\sigma,r)$-minimal codes are necessarily $(\sigma,s)$-minimal for $0\leqslant s\leqslant r$.

\setlength{\parindent}{0em}
\begin{corollary}
{\bf{(1)}}\,\,Let $C\leqslant_{R}X$, and let $D\leqslant_{R}C$ such that $D\subsetneqq C$, $D$ is $\sigma$-maximal in $C$. Then, for any $U\leqslant_{R}D$, it holds that $U$ is $\sigma$-maximal in $C$.

{\bf{(2)}}\,\,Let $C\leqslant_{R}X$, and let $r\in\{0,1,\dots,\dim_{R}(C)-1\}$ such that $C$ is $(\sigma,r)$-minimal. Then, for any $s\in\{0,1,\dots,r\}$, it holds that $C$ is $(\sigma,s)$-minimal.
\end{corollary}

\begin{proof}
{\bf{(1)}}\,\,Let $U\leqslant_{R}D$, and let $V\leqslant_{R}C$ such that $\dim_{R}(V)=\dim_{R}(U)$, $\lambda(U)\preccurlyeq \lambda(V)$. It suffices show that $U=V$. First, we show that for an arbitrary $W\leqslant_{R}D$ such that $U+W=D$, $U\cap W=\{0\}$, it holds that $V+W=D$, $V\cap W=\{0\}$. Indeed, by $\lambda(U)\preccurlyeq \lambda(V)$ and (4) of Lemma 3.1, we have $\lambda(D)=\lambda(U)\curlyvee\lambda(W)\preccurlyeq\lambda(V)\curlyvee\lambda(W)=\lambda(V+W)$. Moreover, we have
\begin{equation}\dim_{R}(D)=\dim_{R}(U)+\dim_{R}(W)=\dim_{R}(V)+\dim_{R}(W)\geqslant\dim_{R}(V+W).\end{equation}
By (3.10), we can choose $P\leqslant_{R}C$ such that $\dim_{R}(P)=\dim_{R}(D)$, $V+W\subseteq P$. Then, we have $\lambda(D)\preccurlyeq\lambda(V+W)\preccurlyeq\lambda(P)$, which, together with the fact that $D$ is $\sigma$-maximal in $C$, implies that $D=P$, which further implies that $V+W\subseteq D$. We claim that $V+W=D$. Indeed, by way of contradiction, suppose that $V+W\subsetneqq D$. Since $D\subsetneqq C$, we can choose $I\leqslant_{R}C$ such that $\dim_{R}(I)=1$, $I\nsubseteq D$. Noticing that $\dim_{R}(V+W+I)\leqslant\dim_{R}(D)$, we can choose $Q\leqslant_{R}C$ such that $\dim_{R}(Q)=\dim_{R}(D)$, $V+W+I\subseteq Q$. Then, we have $\lambda(D)\preccurlyeq \lambda(V+W)\preccurlyeq\lambda(Q)$, which, together with the fact that $D$ is $\sigma$-maximal in $C$, implies that $D=Q$, which further implies that $I\subseteq D$, a contradiction, as desired. It then follows from (3.10) that $\dim_{R}(V+W)=\dim_{R}(V)+\dim_{R}(W)$, which implies that $V\cap W=\{0\}$, as desired. Next, let $L\leqslant_{R}D$ such that $U+L=D$, $U\cap L=\{0\}$. From the first step, we deduce that $V+L=D$, which implies that $V\subseteq D$. Now by way of contradiction, suppose that $V\neq U$. By $\dim_{R}(V)=\dim_{R}(U)$, we have $V\nsubseteq U$. Hence we can choose $I\leqslant_{R}V$ such that $\dim_{R}(I)=1$, $I\cap U=\{0\}$. Since $I\leqslant_{R}D$, we can choose $M\leqslant_{R}D$ such that $U+M=D$, $U\cap M=\{0\}$, $I\subseteq M$. From the first step, we infer that $V\cap M=\{0\}$, a contradiction to $I\subseteq V\cap M$, as desired.

{\bf{(2)}}\,\,Let $B\leqslant_{R}C$ with $\dim_{R}(B)=s$. Then, we can choose $D\leqslant_{R}C$ such that $\dim_{R}(D)=r$, $B\subseteq D$. By Theorem 3.3, $D$ is $\sigma$-maximal in $C$, which, together with (1), implies that $B$ is $\sigma$-maximal in $C$. It then follows from Theorem 3.3 that $C$ is $(\sigma,s)$-minimal, as desired.
\end{proof}

\subsection{Existence results for $(\sigma,r)$-minimal codes}
\setlength{\parindent}{2em}
In this subsection, suppose that $R$ is a finite field with $|R|=q$. For any $(k,r)\in\mathbb{N}^{2}$, define
\begin{equation}\psi(k,r)=|\{C\leqslant_{R}X\mid\dim_{R}(C)=k,\text{$C$ is {\textbf{not}} $(\sigma,r)$-minimal}\}|.\end{equation}
We will derive two sufficient conditions for $k$-dimensional $(\sigma,r)$-minimal codes to exist for $(k,r)\in\mathbb{N}^{2}$. First, an application of Corollary 3.2 yields the following result.

\setlength{\parindent}{0em}
\begin{theorem}
Let $(k,r)\in\mathbb{N}^{2}$ with $r+1\leqslant k\leqslant n$, and let $t\in\{r+1,\dots,k\}$. Then, it holds that \begin{equation}\psi(k,r)\leqslant\psi(t,r)\bin_{q}(n-t,k-t).\end{equation}
Moreover, if
\begin{equation}\psi(t,r)<\prod_{i=k-t+1}^{k}\frac{q^{i+n-k}-1}{q^{i}-1},\end{equation}
then there exists $C\leqslant_{R}X$ such that $\dim_{R}(C)=k$, $C$ is $(\sigma,r)$-minimal.
\end{theorem}

\begin{proof}
We use a combinatorial argument which follows the spirits in the proofs of [17, Theorem 3] and [3, Lemma 6.10]. More precisely, let $T$ denote the following binary relation:
$$T\triangleq\{(A,B)\mid \text{$A,B\leqslant_{R}X$, $A\subseteq B$, $\dim_{R}(A)=t$, $\dim_{R}(B)=k$, $A$ is not $(\sigma,r)$-minimal}\}.$$
From Corollary 3.2 and $r+1\leqslant t\leqslant k\leqslant n$, we deduce that $\dom(T)$ and $\ran(T)$ are equal to the sets of all the non-$(\sigma,r)$-minimal codes of dimension $t$ and dimension $k$, respectively; moreover, for any $A\in\dom(T)$, one can check that $|\{B\mid(A,B)\in T\}|=\bin_{q}(n-t,k-t)$. It then follows that $\psi(k,r)=|\ran(T)|\leqslant|T|=\psi(t,r)\bin_{q}(n-t,k-t)$, which establishes (3.12), as desired. Next, suppose that (3.13) holds. Then, it follows from (3.12) and (2.4) that
$$\psi(k,r)\leqslant\psi(t,r)\bin_{q}(n-t,k-t)<\left(\prod_{i=k-t+1}^{k}\frac{q^{i+n-k}-1}{q^{i}-1}\right)\bin_{q}(n-t,k-t)=\bin_{q}(n,k),$$
which immediately implies the desired result.
\end{proof}

\setlength{\parindent}{2em}
Next, we apply Theorems 3.3 and 3.4 and derive the following result.

\setlength{\parindent}{0em}
\begin{theorem}
Let $r\in\{0,1,\dots,n-1\}$. Then, it holds that
\begin{equation}\psi(r+1,r)\leqslant\frac{1}{q-1}\left(\sum_{s=r+1}^{n}|\{B\leqslant_{R}X\mid\dim_{R}(B)=r,\rk(\lambda(B))=s\}|(q^{s-r}-1)\right).\end{equation}
Moreover, let $k\in\{r+1,\dots,n\}$, and suppose that the right hand side of (3.14) is smaller than $\prod_{i=k-r}^{k}\frac{q^{i+n-k}-1}{q^{i}-1}$. Then, there exists $C\leqslant_{R}X$ such that $\dim_{R}(C)=k$, $C$ is $(\sigma,r)$-minimal.
\end{theorem}

\begin{proof}
Let $T\triangleq\{(B,C)\mid B,C\leqslant_{R}X,\dim_{R}(B)=r,\dim_{R}(C)=r+1,B\subseteq C\subseteq\mu(\lambda(B))\}$. From ``$(1)\Longleftrightarrow(4)$'' of Theorem 3.3 and (1), (2) of Lemma 3.1, one can check that $\ran(T)$ is equal to the set of all the non-$(\sigma,r)$-minimal codes of dimension $r+1$. Moreover, by (3) of Lemma 3.1, we have
\begin{eqnarray*}
\begin{split}
|T|&=\sum_{(B\leqslant_{R}X,\dim_{R}(B)=r)}\frac{q^{\rk(\lambda(B))-r}-1}{q-1}\\
&=\frac{1}{q-1}\left(\sum_{s=r+1}^{n}|\{B\leqslant_{R}X\mid\dim_{R}(B)=r,\rk(\lambda(B))=s\}|(q^{s-r}-1)\right),
\end{split}
\end{eqnarray*}
which, together with $|\ran(T)|\leqslant|T|$, immediately implies (3.14), as desired. Finally, with (3.14), an application of Theorem 3.4 to $t=r+1$ immediately implies the ``moreover'' part.
\end{proof}

\begin{remark}
With $\sigma$ appropriately set, Theorem 3.5 recovers [3, Lemma 6.10] and [13, Theorem 2.12] which have been established for rank metric codes and sum-rank metric codes, respectively. Moreover, in Section 5.3, we will derive the exact value of $\psi(r+1,r)$ for rank metric codes, and then apply Theorem 3.4 to $t=r+1$ to derive an upper bound for the minimal length of $k$-dimensional $r$-minimal rank metric codes (see Theorem 5.4 for more details).
\end{remark}

\section{Cutting $r$-blocking sets and $r$-minimal rank metric codes}
\setlength{\parindent}{2em}
Throughout this section, let $\mathbb{E}/\mathbb{F}$ be a finite dimensional field extension with $\dim_{\mathbb{F}}(\mathbb{E})=m$. We will extensively use the notations introduced in Sections 2.1 and 2.2 without further mentioning.

\subsection{An extension of Lemma 2.2}

\setlength{\parindent}{2em}
In this subsection, we derive the following extension of Lemma 2.2, which is the main result of this subsection and will be crucial for our discussion.

\setlength{\parindent}{0em}
\begin{theorem}
{\bf{(1)}}\,\,Let $Y$ be a finite dimensional $\mathbb{E}$-vector space with $\dim_{\mathbb{E}}(Y)=k$, and let $t\in\{0,1,\dots,k\}$. Then, for any $A\leqslant_{\mathbb{F}} Y$ with $\dim_{\mathbb{F}}(A)\leqslant mt$, there exists $V\leqslant_{\mathbb{E}}Y$ such that $\dim_{\mathbb{E}}(V)=k-t$, $A\cap V=\{0\}$. Moreover, for any $B\leqslant_{\mathbb{F}} Y$ with $\dim_{\mathbb{F}}(B)\geqslant mt$, there exists $W\leqslant_{\mathbb{E}}Y$ such that $\dim_{\mathbb{E}}(W)=k-t$, $B+W=Y$.

{\bf{(2)}}\,\,Let $(n,k)\in\mathbb{N}^{2}$ with $n\geqslant k\geqslant1$, and let $C\leqslant_{\mathbb{E}}\mathbb{E}^{n}$ be an $[n,k]$ code. Then, for $s\in\{0,1,\dots,k\}$, it holds that $\max\{\wt(D)\mid D\leqslant_{\mathbb{E}}C,\dim_{\mathbb{E}}(D)=s\}=\min\{ms,\wt(C)\}$.
\end{theorem}

\begin{proof}
{\bf{(1)}}\,\,First, let $w\in\mathbb{Z}^{+}$, and let $A\leqslant_{\mathbb{F}}\mathbb{E}^{[w]}$ such that $\dim_{\mathbb{F}}(A)\leqslant m$, $\langle A\rangle_{\mathbb{E}}=\mathbb{E}^{[w]}$. We show that there exists $I\leqslant_{\mathbb{E}}\mathbb{E}^{[w]}$ such that $\dim_{\mathbb{E}}(I)=w-1$, $A\cap I=\{0\}$. Indeed, we can choose $G\in\Mat_{w,m}(\mathbb{E})$ such that $\langle\col(G)\rangle_{\mathbb{F}}=A$; moreover, let $C\triangleq\langle \row(G)\rangle_{\mathbb{E}}$. It follows from $\langle A\rangle_{\mathbb{E}}=\mathbb{E}^{[w]}$ that $C$ is an $[m,w]$ code. By Lemma 2.2, we can choose $Q\leqslant_{\mathbb{E}}C$ such that $\dim_{\mathbb{E}}(Q)=1$, $\chi(Q)=\chi(C)$. Let $P\leqslant_{\mathbb{E}}\mathbb{E}^{w}$ such that $Q=\{\gamma G\mid\gamma\in P\}$. By (2) of Lemma 2.1, we have ${P^{\ddagger}}\cap A=\{0\}$; moreover, it follows from $\dim_{\mathbb{E}}(P)=1$ that $\dim_{\mathbb{E}}({P^{\ddagger}})=w-1$, as desired.

\hspace*{4mm}\,\,Now we prove the first assertion via induction on $t$. If $t=0$, then the result trivially holds. Hence we suppose that  $t\geqslant1$, and let $A\leqslant_{\mathbb{F}} Y$ with $\dim_{\mathbb{F}}(A)\leqslant mt$. Then, we can choose $B\leqslant_{\mathbb{F}} Y$ such that $\dim_{\mathbb{F}}(B)=mt$, $A\subseteq B$. Let $H\leqslant_{\mathbb{F}}B$ with $\dim_{\mathbb{F}}(H)=m(t-1)$. By induction, we can choose $M\leqslant_{\mathbb{E}}Y$ such that $\dim_{\mathbb{E}}(M)=k-t+1$, $H\cap M=\{0\}$. This, together with $\dim_{\mathbb{F}}(H)+\dim_{\mathbb{F}}(M)=\dim_{\mathbb{F}}(Y)$, implies that $H+M=Y$. It follows that $B+M=Y$, which further implies that $\dim_{\mathbb{F}}(B\cap M)=m$. Let $X\triangleq\langle B\cap M\rangle_{\mathbb{E}}$. By the discussion in the previous paragraph, we can choose $I\leqslant_{\mathbb{E}}X$ such that $\dim_{\mathbb{E}}(I)=\dim_{\mathbb{E}}(X)-1$, $(B\cap M)\cap I=\{0\}$. Moreover, we can choose $W\leqslant_{\mathbb{E}}Y$ such that $W\cap X=\{0\}$, $W+X=Y$. We infer that $\dim_{\mathbb{E}}(I+W)=k-1$ and $(B\cap M)\cap(I+W)=\{0\}$, which, together with $A\leqslant_{\mathbb{F}}B$, implies that $A\cap(M\cap (I+W))=\{0\}$. Moreover, it is straightforward to verify that $\dim_{\mathbb{E}}(M\cap (I+W))=k-t$, as desired.

\hspace*{4mm}\,\,Finally, we prove the ``moreover'' part. Let $B\leqslant_{\mathbb{F}} Y$ with $\dim_{\mathbb{F}}(B)\geqslant mt$. Then, we can choose $A\leqslant_{\mathbb{F}}B$ with $\dim_{\mathbb{F}}(A)=mt$; moreover, by the first assertion, we can choose $W\leqslant_{\mathbb{E}}Y$ such that $\dim_{\mathbb{E}}(W)=k-t$, $A\cap W=\{0\}$. This, together with $\dim_{\mathbb{F}}(A)+\dim_{\mathbb{F}}(W)=\dim_{\mathbb{F}}(Y)$, implies that $A+W=Y$, which further implies that $B+W=Y$, as desired.

{\bf{(2)}}\,\,Let $G\in\Mat_{k,n}(\mathbb{E})$ be a generator matrix of $C$, and let $U=\langle \col(G)\rangle_{\mathbb{F}}$. By Corollary 2.1, we have $\max\{\wt(D)\mid D\leqslant_{\mathbb{E}}C,\dim_{\mathbb{E}}(D)=s\}\leqslant\min\{ms,\wt(C)\}$. Now we discuss in two cases. First, suppose that $ms\leqslant\dim_{\mathbb{F}}(U)$. By (1), we can choose $V\leqslant_{\mathbb{E}}\mathbb{E}^{[k]}$ such that $\dim_{\mathbb{E}}(V)=k-s$, $U+V=\mathbb{E}^{[k]}$. Let $B\leqslant_{\mathbb{E}}\mathbb{E}^{k}$ such that $V=B^{\ddagger}$, and let $D=\{\gamma G\mid \gamma\in B\}$. By Lemma 2.1, one can check that $\dim_{\mathbb{E}}(D)=\dim_{\mathbb{E}}(B)=s$ and $\wt(D)=\dim_{\mathbb{F}}(\mathbb{E}^{[k]})-\dim_{\mathbb{F}}(V)=ms$, as desired. Next, suppose that $ms\geqslant\dim_{\mathbb{F}}(U)$. By (1), we can choose $W\leqslant_{\mathbb{E}}\mathbb{E}^{[k]}$ such that $\dim_{\mathbb{E}}(W)=k-s$, $U\cap W=\{0\}$. Let $P\leqslant_{\mathbb{E}}\mathbb{E}^{k}$ such that $W=P^{\ddagger}$, and let $Q=\{\gamma G\mid \gamma\in P\}$. By Lemma 2.1, one can check that $\dim_{\mathbb{E}}(Q)=\dim_{\mathbb{E}}(P)=s$ and $\wt(Q)=\wt(C)$, as desired.
\end{proof}

\begin{remark}
Part (2) of Theorem 4.1 recovers Lemma 2.2 and [3, Proposition 3.11] when $s=1$.
\end{remark}

\setlength{\parindent}{2em}
The following proposition improves (1) of Theorem 4.1 for vector spaces over finite fields.

\begin{proposition}
Let $\mathbb{P}$ be a finite field with $|\mathbb{P}|=q$, $X$ be a finite dimensional $\mathbb{P}$-vector space with $\dim_{\mathbb{P}}(X)=k$, $H\subseteq X$ with $0\in H$, and $w=|\{a\in\mathbb{P}-\{0\}\mid aH=H\}|$, where $aH\triangleq\{ah\mid h\in H\}$. Then, for any $t\in\{0,1,\dots,k\}$ such that $|H|\leqslant w(q^{k+1-t}-1)/(q-1)$, there exists $L\leqslant_{\mathbb{P}}X$ such that $\dim_{\mathbb{P}}(L)=t$, $H\cap L=\{0\}$.
\end{proposition}

\begin{proof}
We proceed via induction on $t$. If $t=0$, then the result trivially holds. Hence we suppose that $t\geqslant1$. Then, from $|H|\leqslant w(q^{k}-1)/(q-1)$, we deduce that
$$\left|\bigcup_{a\in \mathbb{P}-\{0\}}aH\right|\leqslant\frac{q-1}{w}|H|\leqslant\frac{q-1}{w}\cdot\frac{w(q^{k}-1)}{q-1}=q^{k}-1=|X|-1.$$
Hence we can choose $z\in X$ such that $z\not\in aH$ for all $a\in\mathbb{P}-\{0\}$. From $0\in H$, we infer that $z\neq0$; moreover, one can check that $I\triangleq\{bz\mid b\in\mathbb{P}\}$ is a $1$-dimensional $\mathbb{P}$-subspace of $X$ with $H\cap I=\{0\}$. Now consider $H+I\subseteq X$, $v\triangleq|\{a\in\mathbb{P}-\{0\}\mid a(H+I)=H+I\}|$ and $t-1$. Since $0\in H$, we have $0\in H+I$; moreover, for any $a\in\mathbb{P}-\{0\}$ with $aH=H$, noticing that $aI=I$, we have $a(H+I)=aH+aI=H+I$. It then follows that $w\leqslant v$, which further implies that
$$|H+I|\leqslant q|H|\leqslant q\cdot\frac{w(q^{k+1-t}-1)}{q-1}\leqslant q\cdot\frac{v(q^{k+1-t}-1)}{q-1}=\frac{v(q^{k+2-t}-q)}{q-1}\leqslant\frac{v(q^{k+1-(t-1)}-1)}{q-1}.$$
By induction, we can choose $J\leqslant_{\mathbb{P}}X$ such that $\dim_{\mathbb{P}}(J)=t-1$, $(H+I)\cap J=\{0\}$. Since $0\in H$, we have $I\subseteq H+I$, which implies that $I\cap J=\{0\}$, which further implies that $\dim_{\mathbb{P}}(I+J)=t$. Moreover, it follows from $H\cap I=\{0\}$ and $(H+I)\cap J=\{0\}$ that $H\cap(I+J)=\{0\}$, as desired.
\end{proof}

\begin{remark}
Proposition 4.1 does not require $H$ to be an additive subgroup of $X$, and hence generalizes (1) of Theorem 4.1 when $\mathbb{E}$ is finite. Thus if $\mathbb{E}$ is finite, then Theorem 4.1 can be alternatively established by using Proposition 4.1. Moreover, if $t=1$, then the restriction $|H|\leqslant w(q^{k}-1)/(q-1)$ is tight. Indeed, by selecting a non-zero element from each $1$-dimensional $\mathbb{P}$-subspace of $X$, together with $0$, one obtains a subset $H\subseteq X$ such that $|H|-1=(q^{k}-1)/(q-1)$.
\end{remark}

\subsection{Necessary and sufficient conditions for linear cutting $r$-blocking sets}

\setlength{\parindent}{2em}
Throughout this subsection, let $Y$ be a finite dimensional $\mathbb{E}$-vector space with $\dim_{\mathbb{E}}(Y)=k\geqslant1$. We will study \textit{linear cutting $r$-blocking sets}, i.e., cutting $r$-blocking sets of $_{\mathbb{E}}Y$ which are $\mathbb{F}$-subspaces. Following [3, Section 6.4], for any $A\leqslant_{\mathbb{F}}Y$, the \textit{linearity index} of $A$, denoted by $\mathcal{L}(A)$, is defined as
\begin{equation}\mathcal{L}(A)\triangleq\max\{\dim_{\mathbb{E}}(V)\mid V\leqslant_{\mathbb{E}}Y,V\subseteq A\}.\end{equation}

The following theorem is the main result of this subsection, in which we characterize linear cutting $r$-blocking sets in terms of evasive subspaces. When $\mathbb{E}$ is finite, the counterpart result has been established in [8, Theorem 3.3] for $r=1$ and in [31, Theorem 2.3] for arbitrary $r$ by using a counting argument.

\setlength{\parindent}{0em}
\begin{theorem}
Let $A\leqslant_{\mathbb{F}}Y$, and let $r\in\{0,1,\dots,k-1\}$. Then, the following three statements are equivalent to each other:

{\bf{(1)}}\,\,$A$ is a cutting $r$-blocking set of $Y$;

{\bf{(2)}}\,\,For any $W\leqslant_{\mathbb{E}} Y$ with $\dim_{\mathbb{E}}(W)=k-r-1$, it holds that $\dim_{\mathbb{F}}(A+W)\geqslant(k-1)m+1$;

{\bf{(3)}}\,\,$A$ is $(k-r-1,\dim_{\mathbb{F}}(A)-mr-1)$-evasive in $Y$.
\end{theorem}

\begin{proof}
Since (1) implies that $\langle A\rangle_{\mathbb{E}}=Y$, with Definition 2.3 and some straightforward computation, we deudce that $(3)\Longrightarrow(2)$ and $((1)\wedge(2))\Longrightarrow(3)$. Therefore it suffices to establish $(1)\Longleftrightarrow(2)$. First, suppose that (1) holds. Let $W\leqslant_{\mathbb{E}} Y$ with $\dim_{\mathbb{E}}(W)=k-r-1$. From Proposition 2.1, we deduce that $(A+W)\cap I\neq\{0\}$ for all $I\leqslant_{\mathbb{E}} Y$ with $\dim_{\mathbb{E}}(I)=1$. This, together with (1) of Theorem 4.1, implies that $\dim_{\mathbb{F}}(A+W)\geqslant(k-1)m+1$, which further establishes (2), as desired. Second, suppose that (2) holds. Let $W,I\leqslant_{\mathbb{E}} Y$ such that $\dim_{\mathbb{E}}(W)=k-r-1$, $\dim_{\mathbb{E}}(I)=1$. By (2), we have $\dim_{\mathbb{F}}(A+W)+\dim_{\mathbb{F}}(I)\geqslant\dim_{\mathbb{F}}(Y)+1$, which implies that $(A+W)\cap I\neq\{0\}$. It then follows from Proposition 2.1 that $A$ is a cutting $r$-blocking set of $Y$, as desired.
\end{proof}

\setlength{\parindent}{2em}
Next, we give a lower bound for the linearity index. When $\mathbb{E}$ is finite, the following result has been established in [3, Proposition 6.15] by using different methods.

\setlength{\parindent}{0em}
\begin{proposition}
Let $s\in\{0,1,\dots,k\}$, and let $B\leqslant_{\mathbb{F}}Y$ with $\dim_{\mathbb{F}}(B)\geqslant k(m-1)+s$. Then, it holds that $\mathcal{L}(B)\geqslant s$.
\end{proposition}

\begin{proof}
Let $g:Y\times Y\longrightarrow \mathbb{E}$ be a non-degenerated symmetric $\mathbb{E}$-bilinear map, $\tau:\mathbb{E}\longrightarrow \mathbb{F}$ be a non-zero $\mathbb{F}$-linear map. Then, $\tau\circ g:Y\times Y\longrightarrow \mathbb{F}$ is a non-degenerated symmetric $\mathbb{F}$-bilinear map. For any $A\subseteq Y$, set $A^{\star}\triangleq\{v\in Y\mid \text{$(\tau\circ g)(u,v)=0$ for all $u\in A$}\}$. Now for $J\triangleq\langle B^{\star}\rangle_{\mathbb{E}}$, we have
\begin{equation}\dim_{\mathbb{E}}(J)\leqslant\dim_{\mathbb{F}}(B^{\star})=\dim_{\mathbb{F}}(Y)-\dim_{\mathbb{F}}(B)\leqslant km-(k(m-1)+s)=k-s.\end{equation}
By $J\leqslant_{\mathbb{E}}Y$, $B^{\star}\subseteq J$, we have $J^{\star}\leqslant_{\mathbb{E}}Y$, $J^{\star}\subseteq B$; moreover, it follows from (4.2) that $\dim_{\mathbb{E}}(J^{\star})=k-\dim_{\mathbb{E}}(J)\geqslant s$, which further establishes the desired result.
\end{proof}

\setlength{\parindent}{2em}
Now we give two corollaries of Theorem 4.2 and Proposition 4.2. The following result will be used frequently in Section 5 to establish lower bounds for the dimensions of cutting $r$-blocking sets.

\setlength{\parindent}{0em}
\begin{corollary}
Fix $r\in\{0,1,\dots,k-1\}$. Let $b\in\mathbb{N}$ with $b\leqslant km$. Suppose that for any $J\leqslant_{\mathbb{F}}Y$ such that $J$ is $(k-r-1,b-mr-1)$-evasive in $Y$, it holds that $\dim_{\mathbb{F}}(J)\leqslant b-1$. Then, for any $A\leqslant_{\mathbb{F}}Y$ such that $A$ is a cutting $r$-blocking set of $Y$, it hold that $\dim_{\mathbb{F}}(A)\geqslant b+1$.
\end{corollary}

\begin{proof}
By way of contradiction, suppose that $A\leqslant_{\mathbb{F}}Y$ is a cutting $r$-blocking set of $Y$ with $\dim_{\mathbb{F}}(A)\leqslant b$. Since $b\leqslant km$, we can choose $B\leqslant_{\mathbb{F}}Y$ such that $A\subseteq B$, $\dim_{\mathbb{F}}(B)=b$. Then, $B$ is a cutting $r$-blocking set of $Y$, which, together with Theorem 4.2, implies that $B$ is $(k-r-1,b-mr-1)$-evasive in $Y$, which further implies that $\dim_{\mathbb{F}}(B)\leqslant b-1$, a contradiction, as desired.
\end{proof}

\setlength{\parindent}{0em}
\begin{corollary}
{\bf{(1)}}\,\,Let $A\leqslant_{\mathbb{F}}Y$. Then, $A$ is a cutting $(k-1)$-blocking set of $Y$ if and only if $\dim_{\mathbb{F}}(A)\geqslant(k-1)m+1$, if and only if $A$ is a cutting $r$-blocking set of $Y$ for all $r\in\{0,1,\dots,k-1\}$.

{\bf{(2)}}\,\,Let $A\leqslant_{\mathbb{F}}Y$ with $\dim_{\mathbb{F}}(A)=(k-1)m$, and let $r\in\{0,1,\dots,k-1\}$. Then, $A$ is a cutting $r$-blocking set of $Y$ if and only if $\mathcal{L}(A)\leqslant k-r-2$.

{\bf{(3)}}\,\,Let $r\in\{0,1,\dots,k-1\}$, and let $A\leqslant_{\mathbb{F}}Y$ be a cutting $r$-blocking set of $Y$. Then, for $s\triangleq\min\{k-r-1,\mathcal{L}(A)\}$, it holds that $\dim_{\mathbb{F}}(A)\geqslant(m-1)(r+s)+k$.

{\bf{(4)}}\,\,Suppose that $k\geqslant m$. Let $r\in\mathbb{N}$ such that $m-1\leqslant r\leqslant k-1$, and let $A\leqslant_{\mathbb{F}}Y$ be a cutting $r$-blocking set of $Y$. Then, it holds that $\dim_{\mathbb{F}}(A)\geqslant (k-1)m+1$.
\end{corollary}

\begin{proof}
{\bf{(1)}}\,\,This immediately follows from ``$(1)\Longleftrightarrow(2)$'' of Theorem 4.2.

{\bf{(2)}}\,\,From Theorem 4.2, we deduce that $A$ is a cutting $r$-blocking set of $Y$ if and only if $W\nsubseteq A$ for all $W\leqslant_{\mathbb{E}} Y$ with $\dim_{\mathbb{E}}(W)=k-r-1$, if and only if $\mathcal{L}(A)\leqslant k-r-2$, as desired.

{\bf{(3)}}\,\,Let $J\leqslant_{\mathbb{E}}Y$ such that $\dim_{\mathbb{E}}(J)=s$, $J\subseteq A$. Since $\langle A\rangle_{\mathbb{E}}=Y$, $s\leqslant k-r-1\leqslant k$, we can choose $S\subseteq A$ such that $|S|=k-r-1-s$, $S$ is linear independent over $\mathbb{E}$, and $\langle S\rangle_{\mathbb{E}}\cap J=\{0\}$. Then, for $W\triangleq\langle S\rangle_{\mathbb{E}}+J$, we have $\dim_{\mathbb{E}}(W)=|S|+s=k-r-1$, which, together with Theorem 4.2, implies that $\dim_{\mathbb{F}}(A\cap W)\leqslant \dim_{\mathbb{F}}(A)-mr-1$. This, together with $J+\langle S\rangle_{\mathbb{F}}\subseteq A\cap W$, further implies that $\dim_{\mathbb{F}}(J+\langle S\rangle_{\mathbb{F}})\leqslant \dim_{\mathbb{F}}(A)-mr-1$. Noticing that $\dim_{\mathbb{F}}(J+\langle S\rangle_{\mathbb{F}})=ms+|S|=(m-1)s+k-r-1$, the desired result follows from some straightforward verification which we omit.

{\bf{(4)}}\,\,By way of contradiction, suppose that $\dim_{\mathbb{F}}(A)\leqslant (k-1)m$. Then, we can choose $B\leqslant_{\mathbb{F}}Y$ such that $\dim_{\mathbb{F}}(B)=(k-1)m$, $A\subseteq B$. Noticing that $B$ is a cutting $r$-blocking set of $Y$, from (2), we deudce that $\mathcal{L}(B)\leqslant k-r-2\leqslant k-m-1$. However, Proposition 4.3 implies that $\mathcal{L}(B)\geqslant k-m$, a contradiction, as desired.
\end{proof}

\subsection{Necessary and sufficient conditions for $r$-minimal codes}
\setlength{\parindent}{2em}
Throughout this subsection, let $(n,k)\in\mathbb{N}^{2}$ with $n\geqslant k\geqslant1$. We will give several necessary and sufficient conditions for an $[n,k]$ code to be $r$-minimal$^1$. The following theorem is the first main result of this subsection, in which we characterize rank minimal subcodes and $r$-minimal codes in terms of generator matrices and cutting $r$-blocking sets.

\renewcommand{\thefootnote}{\fnsymbol{footnote}}

\footnotetext{\hspace*{-6mm} \begin{tabular}{@{}r@{}p{16cm}@{}}
$^1$ When it is clear from the context, we will simply say ``$r$-minimal'' and omit the term ``with respect to rank metric''.
\end{tabular}}

\setlength{\parindent}{0em}
\begin{theorem}
Let $C\leqslant_{\mathbb{E}}\mathbb{E}^{n}$ be an $[n,k]$ code, $G\in\Mat_{k,n}(\mathbb{E})$ be a generator matrix of $C$, and let $U=\langle\col(G)\rangle_{\mathbb{F}}\leqslant_{\mathbb{F}}\mathbb{E}^{[k]}$. Then, for $B\leqslant_{\mathbb{E}}\mathbb{E}^{k}$ and $D\triangleq\{\gamma G\mid\gamma\in B\}$, $D$ is rank minimal in $C$ if and only if $\langle{B^{\ddagger}}\cap U\rangle_{\mathbb{E}}=B^{\ddagger}$. Moreover, for $r\in\{0,1,\dots,k\}$, $C$ is $r$-minimal if and only if $U$ is a cutting $r$-blocking set of $\mathbb{E}^{[k]}$.
\end{theorem}

\begin{proof}
First, suppose that $\langle{B^{\ddagger}}\cap U\rangle_{\mathbb{E}}\subsetneqq B^{\ddagger}$. Then, from $\langle U\rangle_{\mathbb{E}}=\mathbb{E}^{[k]}$, we deduce that $B^{\ddagger}\subsetneqq\mathbb{E}^{[k]}$. Hence we can choose $I\leqslant_{\mathbb{E}}\mathbb{E}^{[k]}$ such that $\dim_{\mathbb{E}}(I)=1$, $I\nsubseteq B^{\ddagger}$. Moreover, we can choose $W\leqslant_{\mathbb{E}}\mathbb{E}^{[k]}$ such that $\dim_{\mathbb{E}}(W)=\dim_{\mathbb{E}}(B^{\ddagger})$, $\langle{B^{\ddagger}}\cap U\rangle_{\mathbb{E}}+I\subseteq W$. Let $P\leqslant_{\mathbb{E}}\mathbb{E}^{k}$ such that $W=P^{\ddagger}$, and let $Q=\{\gamma G\mid\gamma\in P\}$. Then, we have $\dim_{\mathbb{E}}(Q)=\dim_{\mathbb{E}}(P)=\dim_{\mathbb{E}}(B)=\dim_{\mathbb{E}}(D)$; moreover, by Lemma 2.1, we have $\chi(Q)\subseteq\chi(D)$. In addition, by $B^{\ddagger}\neq W=P^{\ddagger}$, we have $B\neq P$, which implies that $D\neq Q$. It then follows from Theorem 3.2 that $D$ is not rank minimal in $C$, as desired.

\hspace*{4mm}\,\,Second, suppose that $\langle{B^{\ddagger}}\cap U\rangle_{\mathbb{E}}=B^{\ddagger}$. Consider $Q\leqslant_{\mathbb{E}}C$ such that $\dim_{\mathbb{E}}(Q)=\dim_{\mathbb{E}}(D)$, $\chi(Q)\subseteq\chi(D)$. Let $P\leqslant_{\mathbb{E}}\mathbb{E}^{k}$ such that $Q=\{\gamma G\mid\gamma\in P\}$. By Lemma 2.1, we have ${B^{\ddagger}}\cap U\subseteq P^{\ddagger}$, which implies that $B^{\ddagger}=\langle{B^{\ddagger}}\cap U\rangle_{\mathbb{E}}\subseteq P^{\ddagger}$. It follows that $P\subseteq B$, which implies that $Q\subseteq D$. This, together with $\dim_{\mathbb{E}}(Q)=\dim_{\mathbb{E}}(D)$, further implies that $D=Q$, as desired.

\hspace*{4mm}\,\,Finally, from the first assertion, we deduce that $C$ is $r$-minimal if and only if $\langle{A^{\ddagger}}\cap U\rangle_{\mathbb{E}}=A^{\ddagger}$ for all $A\leqslant_{\mathbb{E}}\mathbb{E}^{k}$ with $\dim_{\mathbb{E}}(A)=r$, if and only if $\langle W\cap U\rangle_{\mathbb{E}}=W$ for all $W\leqslant_{\mathbb{E}}\mathbb{E}^{[k]}$ with $\dim_{\mathbb{E}}(W)=k-r$, if and only if $U$ is a cutting $r$-blocking set of $\mathbb{E}^{[k]}$, as desired.
\end{proof}

\setlength{\parindent}{2em}
Theorem 4.3 naturally extends to sum-rank metric, as detailed in the following remark.

\begin{remark}
Following the notations in Example 3.3, let $C\leqslant_{\mathbb{E}}\prod_{i=1}^{t}\mathbb{E}^{w_{i}}$ with $\dim_{\mathbb{E}}(C)=k\geqslant1$, and let $(H_1,\dots,H_t)\in\prod_{s=1}^{t}\Mat_{k,w_s}(\mathbb{E})$ such that $C=\{(\gamma H_1,\dots,\gamma H_t)\mid\gamma\in \mathbb{E}^{k}\}$; moreover, for any $s\in\{1,\dots,t\}$, let $U_{s}\triangleq\langle\col(H_s)\rangle_{\mathbb{F}}$. Then, for $B\leqslant_{\mathbb{E}}\mathbb{E}^{k}$ and $D\triangleq\{(\gamma H_1,\dots,\gamma H_t)\mid\gamma\in B\}$, with the help of Lemma 2.1 and Theorem 3.2, one can show that $D$ is sum-rank minimal in $C$ if and only if $\langle{B^{\ddagger}}\cap (\bigcup_{s=1}^{t}U_s)\rangle_{\mathbb{E}}=B^{\ddagger}$; moreover, for $r\in\{0,1,\dots,k\}$, $C$ is $r$-minimal with respect to sum-rank metric if and only if $\bigcup_{s=1}^{t}U_s$ is a cutting $r$-blocking set of $\mathbb{E}^{[k]}$. This result recovers Theorem 4.3 when $t=1$, and recovers [44, Corollary 9.24] when $r=1$ and $\mathbb{E}$ is finite.
\end{remark}

\setlength{\parindent}{2em}
Now we state and prove the second main result of this subsection. In the following, we characterize $r$-minimal codes in terms of generalized rank weights of the codes via two proofs, where the first one relies on Theorems 4.2 and 4.3, and the second one relies on Theorems 3.3 and 4.1 and does not involve cutting $r$-blocking sets.

\setlength{\parindent}{0em}
\begin{theorem}
Let $C\leqslant_{\mathbb{E}}\mathbb{E}^{n}$ be an $[n,k]$ code, and let $r\in\{0,1,\dots,k-1\}$. Then, $C$ is $r$-minimal if and only if $\mathbf{d}_{r+1}(C)\geqslant mr+1$.
\end{theorem}

\begin{proof}[First proof]
Let $G\in\Mat_{k,n}(\mathbb{E})$ be a generator matrix of $C$, and let $U=\langle \col(G)\rangle_{\mathbb{F}}$. From either [39, Theorems 2 and 4] or Lemma 2.1, we deduce that
$$\mathbf{d}_{r+1}(C)=\min\{\dim_{\mathbb{F}}(U)-\dim_{\mathbb{F}}(M\cap U)\mid M\leqslant_{\mathbb{E}}\mathbb{E}^{[k]},\dim_{\mathbb{E}}(M)=k-r-1\}.$$
Hence from Theorems 4.2 and 4.3, we deduce that $\mathbf{d}_{r+1}(C)\geqslant mr+1$ if and only if $U$ is a cutting $r$-blocking set of $\mathbb{E}^{[k]}$, if and only if $C$ is $r$-minimal, as desired.
\end{proof}

\begin{proof}[Second proof]
First, suppose that $\mathbf{d}_{r+1}(C)\geqslant mr+1$. Then, for any $D,P\leqslant_{\mathbb{E}}C$ such that $\dim_{\mathbb{E}}(P)=r$, $\dim_{\mathbb{E}}(D)=r+1$, we have $\wt(P)\leqslant mr<\mathbf{d}_{r+1}(C)\leqslant\wt(D)$. This, together with ``$(1)\Longleftrightarrow(4)$'' of Theorem 3.3, immediately implies that $C$ is $r$-minimal, as desired. Next, suppose that $C$ is $r$-minimal. It suffices to show that $\wt(D)\geqslant mr+1$ for an arbitrary $D\leqslant_{\mathbb{E}}C$ with $\dim_{\mathbb{E}}(D)=r+1$. Indeed, if $\wt(D)\leqslant mr$, then by (2) of Theorem 4.1, there exists $Q\leqslant_{\mathbb{E}}D$ such that $\dim_{\mathbb{E}}(Q)=r$, $\chi(Q)=\chi(D)$, a contradiction to ``$(1)\Longleftrightarrow(4)$'' of Theorem 3.3, as desired.
\end{proof}

\setlength{\parindent}{2em}
Theorems 4.3 and 4.4, along with Lemma 2.1, immediately imply the following coding-theoretic counterpart of Corollary 4.2.

\setlength{\parindent}{0em}
\begin{corollary}
Let $C\leqslant_{\mathbb{E}}\mathbb{E}^{n}$ be an $[n,k]$ code, $G\in\Mat_{k,n}(\mathbb{E})$ be a generator matrix of $C$, $U=\langle \col(G)\rangle_{\mathbb{F}}$, and $r\in\{0,1,\dots,k-1\}$. Then, the following four statements hold:

{\bf{(1)}}\,\,If $\wt(C)\geqslant(k-1)m+1$, then $C$ is $r$-minimal;

{\bf{(2)}}\,\,If $\wt(C)=(k-1)m$, then $C$ is $r$-minimal if and only if $\mathcal{L}(U)\leqslant k-r-2$;

{\bf{(3)}}\,\,If $C$ is $r$-minimal, then for $s\triangleq\min\{k-r-1,\mathcal{L}(U)\}$, it holds that $\wt(C)\geqslant (m-1)(r+s)+k$;

{\bf{(4)}}\,\,If $r\geqslant m-1$ and $C$ is $r$-minimal, then it holds that $\wt(C)\geqslant (k-1)m+1$.
\end{corollary}

\setlength{\parindent}{2em}
Now we characterize $r$-minimal codes in terms of generalized rank weights of the dual codes.

\setlength{\parindent}{0em}
\begin{theorem}
Suppose that $m\geqslant2$. Let $C\leqslant_{\mathbb{E}}\mathbb{E}^{n}$ be an $[n,k]$ code, and let $r\in\{1,\dots,k-1\}$ such that $n\geqslant(m-1)r+k$. Then, $C$ is $r$-minimal if and only if $\mathbf{d}_{n-(m-1)r-k+1}(C^{\bot})\geqslant n-mr+1$.
\end{theorem}

\begin{proof}
We recall a general fact which follows from [33, Theorem 3.3] and can be derived by using Wei-type duality (see [23, Theorem I.3]). More precisely, for $(c,s)\in\mathbb{N}^{2}$ such that $\max\{0,s-k\}\leqslant c\leqslant\min\{n-k,s\}$, it holds that $\mathbf{d}_{c}(C^{\bot})\geqslant s+1$ if and only if $\mathbf{d}_{c+k-s}(C)\geqslant n-s+1$. An application of this fact to $(n-(m-1)r-k+1,n-mr)$ implies that $\mathbf{d}_{n-(m-1)r-k+1}(C^{\bot})\geqslant n-mr+1$ if and only if $\mathbf{d}_{r+1}(C)\geqslant mr+1$, which, together with Theorem 4.4, further implies the desired result.
\end{proof}

\setlength{\parindent}{2em}
Finally, we characterize codes whose $r$-dimensional subcodes have constant rank support weight (such codes are necessarily $r$-minimal by Proposition 3.1). When $r=1$ and $\mathbb{E}$ is finite, these codes were classified in \cite{39,3} under the name of Hadamard code and simplex code, respectively. Now we further extend the results in \cite{39,3} to arbitrary $r$ and possibly infinite $\mathbb{E}$ via a different approach.

\setlength{\parindent}{0em}
\begin{theorem}
Let $C$ be an $[n,k]$ code, $G\in\Mat_{k,n}(\mathbb{E})$ be a generator matrix of $C$, $U=\langle \col(G)\rangle_{\mathbb{F}}$, and $r\in\{1,\dots,k-1\}$. Then, the following three statements are equivalent to each other:

{\bf{(1)}}\,\,For any $D,L\leqslant_{\mathbb{E}}C$ with $\dim_{\mathbb{E}}(D)=\dim_{\mathbb{E}}(L)=r$, it holds that $\wt(D)=\wt(L)$;

{\bf{(2)}}\,\,$U=\mathbb{E}^{[k]}$;

{\bf{(3)}}\,\,$\wt(C)=\dim_{\mathbb{F}}(C)$.
\end{theorem}

\begin{proof}
Noticing that both $(2)\Longleftrightarrow(3)$ and $(3)\Longrightarrow(1)$ follow from Corollary 2.1, we only prove $(1)\Longrightarrow(2)$. First, we show that $\dim_{\mathbb{F}}(V\cap U)=\dim_{\mathbb{F}}(W\cap U)$ for $V,W\leqslant_{\mathbb{E}}\mathbb{E}^{[k]}$ with $\dim_{\mathbb{E}}(V)=\dim_{\mathbb{E}}(W)=k-r$. Indeed, let $B,P\leqslant_{\mathbb{E}}\mathbb{E}^{k}$ such that $V=B^{\ddagger}$, $W=P^{\ddagger}$, and let $D=\{\gamma G\mid\gamma\in B\}$, $Q=\{\gamma G\mid\gamma\in P\}$. Noticing that $\dim_{\mathbb{E}}(D)=\dim_{\mathbb{E}}(Q)=r$, we have $\wt(D)=\wt(Q)$, which, together with Lemma 2.1, implies that $\dim_{\mathbb{F}}(V\cap U)=\dim_{\mathbb{F}}(W\cap U)$, as desired. It follows that $U\cap L\neq\{0\}$ for all $L\leqslant_{\mathbb{E}}\mathbb{E}^{[k]}$ with $\dim_{\mathbb{E}}(L)=k-r$. By Theorem 4.1, we have $\dim_{\mathbb{F}}(U)\geqslant mr+1$, and there exists $V\leqslant_{\mathbb{E}}\mathbb{E}^{[k]}$ such that $\dim_{\mathbb{E}}(V)=k-r$, $U+V=\mathbb{E}^{[k]}$. Now suppose that $U\neq\mathbb{E}^{[k]}$. Then, we can choose $M\leqslant_{\mathbb{F}}\mathbb{E}^{[k]}$ such that $\dim_{\mathbb{F}}(M)=km-1$, $U\subseteq M$. By Proposition 4.3, we can choose $W\leqslant_{\mathbb{E}}\mathbb{E}^{[k]}$ such that $\dim_{\mathbb{E}}(W)=k-r$, $W\subseteq M$. From $\dim_{\mathbb{F}}(U\cap W)=\dim_{\mathbb{F}}(U\cap V)$, we deduce that $\dim_{\mathbb{F}}(U+W)=\dim_{\mathbb{F}}(U+V)=km$, a contradiction to $U+W\subseteq M$, as desired.
\end{proof}

\setlength{\parindent}{2em}
\begin{remark}
When $r=1$ and $\mathbb{E}$ is finite, Theorem 4.3 recovers [3, Corollary 5.7], Theorem 4.4 recovers [8, Theorem 3.4], Theorem 4.6 recovers [39, Theorem 12], [3, Corollary 3.17] and part of [3, Proposition 3.16], and Corollary 4.3 recovers Proposition 6.2, Corollary 6.19, Proposition 6.17 and the ``only if'' part of Corollary 6.20 of \cite{3}, respectively. Unlike Theorem 4.3, Theorems 4.4 and 4.5 do not seem to have analogues with respect to Hamming metric or sum-rank metric.
\end{remark}

\section{Minimal length of $r$-minimal codes}
\setlength{\parindent}{2em}
Throughout this section, let $\mathbb{E}/\mathbb{F}$ be a finite dimensional field extension with $\dim_{\mathbb{F}}(\mathbb{E})=m$ as in Section 4. For any $(k,r)\in\mathbb{N}^{2}$ with $k\geqslant r+1$, define
\begin{equation}\varpi_{\mathbb{E}/\mathbb{F}}(k,r)\triangleq\min\{\dim_{\mathbb{F}}(U)\mid \text{$U\leqslant_{\mathbb{F}}\mathbb{E}^{[k]}$, $U$ is a cutting $r$-blocking set of $\mathbb{E}^{[k]}$}\},\end{equation}
which in fact is also equal to the minimal length of all the $k$-dimensional $r$-minimal codes, as detailed in the following proposition.

\setlength{\parindent}{0em}
\begin{proposition}
Let $(k,r)\in\mathbb{N}^{2}$ with $k\geqslant r+1$. Then, for any $n\in\mathbb{Z}^{+}$ and any $[n,k]$ $r$-minimal code $C\leqslant_{\mathbb{E}}\mathbb{E}^{n}$, it holds that $n\geqslant\wt(C)\geqslant\varpi_{\mathbb{E}/\mathbb{F}}(k,r)$. Conversely, for any $n\in\mathbb{Z}^{+}$ with $n\geqslant\varpi_{\mathbb{E}/\mathbb{F}}(k,r)$, there exists an $[n,k]$ $r$-minimal code. In particular, it holds that
\begin{equation}\varpi_{\mathbb{E}/\mathbb{F}}(k,r)=\min\{n\in\mathbb{Z}^{+}\mid\text{there exists an $[n,k]$ $r$-minimal code}\}.\end{equation}
\end{proposition}

\begin{proof}
First, let $n\in\mathbb{Z}^{+}$, $C\leqslant_{\mathbb{E}}\mathbb{E}^{n}$ be an $[n,k]$ $r$-minimal code, $G\in\Mat_{k,n}(\mathbb{E})$ be a generator matrix of $C$, and let $U=\langle \col(G)\rangle_{\mathbb{F}}$. By Theorem 4.3, $U$ is a cutting $r$-blocking set of $\mathbb{E}^{[k]}$. It then follows from Lemma 2.1 and (5.1) that $n\geqslant\wt(C)=\dim_{\mathbb{F}}(U)\geqslant\varpi_{\mathbb{E}/\mathbb{F}}(k,r)$, as desired.

\hspace*{4mm}\,\,Second, let $n\in\mathbb{Z}^{+}$ with $\varpi_{\mathbb{E}/\mathbb{F}}(k,r)\leqslant n$. Then, we can choose $G\in\Mat_{k,n}(\mathbb{E})$ such that $U\triangleq\langle \col(G)\rangle_{\mathbb{F}}$ is a cutting $r$-blocking set of $\mathbb{E}^{[k]}$. From $\langle U\rangle_{\mathbb{E}}=\mathbb{E}^{[k]}$ and Theorem 4.3, we deduce that $\rank(G)=k\leqslant n$, and $C\triangleq\langle \row(G)\rangle_{\mathbb{E}}$ is an $[n,k]$ $r$-minimal code, as desired.
\end{proof}

\setlength{\parindent}{2em}
The rest of this section is devoted to deriving lower and upper bounds for $\varpi_{\mathbb{E}/\mathbb{F}}(k,r)$ as well as giving exact values of $\varpi_{\mathbb{E}/\mathbb{F}}(k,r)$ for some special scenarios. When $\mathbb{E}$ is finite, $\varpi_{\mathbb{E}/\mathbb{F}}(k,1)$ has been extensively studied in terms of both minimal codes and cutting blocking sets, and we first recall some known results in the literature.

\setlength{\parindent}{0em}
\begin{lemma}
Suppose that $\mathbb{E}$ is finite. Then, the following eight statements hold:

{\bf{(1)}}\,\,For any $k\in\mathbb{N}$ with $k\geqslant 2$, it holds that $\varpi_{\mathbb{E}/\mathbb{F}}(k,1)\geqslant m+k-1$ ([3, Corollary 5.10]);

{\bf{(2)}}\,\,For any $k\in\mathbb{N}$ with $k\geqslant\lfloor\sqrt{m}\rfloor+2$, it holds that $\varpi_{\mathbb{E}/\mathbb{F}}(k,1)\geqslant m+k$ ([8, Corollary 3.5]);

{\bf{(3)}}\,\,If $m=3$, then $\varpi_{\mathbb{E}/\mathbb{F}}(4,1)=8$ ([8, Section 3.4]);

{\bf{(4)}}\,\,If $m\geqslant4$, then $\varpi_{\mathbb{E}/\mathbb{F}}(3,1)\leqslant m+3$ ([26, Theorem 7.16]);

{\bf{(5)}}\,\,For any $k\in\mathbb{N}$ with $k\geqslant 2$, it holds that $\varpi_{\mathbb{E}/\mathbb{F}}(k,1)\leqslant m+2k-2$ ([3, Theorem 6.11]);

{\bf{(6)}}\,\,If $m=4$ and $|\mathbb{F}|$ is an odd power of $2$, then $\varpi_{\mathbb{E}/\mathbb{F}}(4,1)=8$ ([8, Proposition 4.9]);

{\bf{(7)}}\,\,If $m\geqslant4$ and $m\not\equiv3,5~(\bmod~6)$, then $\varpi_{\mathbb{E}/\mathbb{F}}(3,1)=m+2$ ([3, Theorem 6.7]);

{\bf{(8)}}\,\,If $m\geqslant5$ and $2\nmid m$, then $\varpi_{\mathbb{E}/\mathbb{F}}(3,1)=m+2$ with some additional assumptions on $\mathbb{E}/\mathbb{F}$ ([9, Theorem 5.1] and [31, Theorem 3.5]).
\end{lemma}

\setlength{\parindent}{2em}
In this section, among others, we will further generalize (1)--(5) of Lemma 5.1. First of all, we give a lower bound and an upper bound for $\varpi_{\mathbb{E}/\mathbb{F}}(k,r)$, both of which are tight for some special cases. The following result immediately follows from Corollary 4.2.

\setlength{\parindent}{0em}
\begin{corollary}
{\bf{(1)}}\,\,For any $(k,r)\in\mathbb{N}^{2}$ with $k\geqslant r+1$, it holds that
\begin{equation}(m-1)r+k\leqslant\varpi_{\mathbb{E}/\mathbb{F}}(k,r)\leqslant(k-1)m+1.\end{equation}
{\bf{(2)}}\,\,For any $(k,r)\in\mathbb{N}^{2}$ with $k\geqslant r+1\geqslant m$, it holds that
$\varpi_{\mathbb{E}/\mathbb{F}}(k,r)=(k-1)m+1$.

{\bf{(3)}}\,\,For any $k\in\mathbb{Z}^{+}$, it holds that $\varpi_{\mathbb{E}/\mathbb{F}}(k,k-1)=(k-1)m+1$, $\varpi_{\mathbb{E}/\mathbb{F}}(k,0)=k$.
\end{corollary}

\begin{remark}
When $r=1$ and $\mathbb{E}$ is finite, the lower bound in (5.3) recovers (1) of Lemma 5.1. We will show in Sections 5.2 and 5.3 that the two bounds in (5.3) are not tight in general.
\end{remark}

\subsection{Evasive subspaces}

\setlength{\parindent}{2em}
For further discussion, in this subsection, we derive some upper bounds for the dimensions of evasive subspaces. The following notation will be crucial for our discussion.

\setlength{\parindent}{0em}
\begin{notation}
{\bf{(1)}}\,\,For $(\lambda,a,u)\in\mathbb{N}^{3}$ and $k\in\mathbb{N}$, write $k\leftrightarrow(\lambda,a,u)$ if $k\geqslant\lambda$, and moreover, for any $n\in\mathbb{N}$ with $n\geqslant k$, and any $U\leqslant_{\mathbb{F}}\mathbb{E}^{n}$ such that $U$ is $(n-\lambda,n-\lambda+a)$-evasive in $\mathbb{E}^{n}$, it holds that $\dim_{\mathbb{F}}(U)\leqslant n+a+u$.

{\bf{(2)}}\,\,For $(\lambda,a,u)\in\mathbb{N}^{3}$, $(\varepsilon,b,v)\in\mathbb{N}^{3}$, write $(\lambda,a,u)\curlyeqprec(\varepsilon,b,v)$ if for any $k\in\mathbb{N}$ with $k\leftrightarrow(\lambda,a,u)$, it holds that $k\leftrightarrow(\varepsilon,b,v)$.
\end{notation}

\setlength{\parindent}{2em}
We begin with the following proposition, which is inspired by [8, Lemma 3.11].

\setlength{\parindent}{0em}
\begin{proposition}
{\bf{(1)}}\,\,Let $X$ be a $k$-dimensional $\mathbb{E}$-vector space, $(a,v)\in\mathbb{N}^{2}$, $(b,w)\in\mathbb{N}^{2}$ with $a\leqslant b\leqslant k$, and $U\leqslant_{\mathbb{F}}X$ be both $(a,v)$-evasive and $(b,w)$-evasive in $X$. Moreover, let $g_1\in\mathbb{R}$ such that $\dim_{\mathbb{F}}(I)\leqslant g_1$ for all $(a,v-1)$-evasive $\mathbb{F}$-subspace $I$ in $\mathbb{E}^{k}$, and let $g_2\in\mathbb{R}$ such that $\dim_{\mathbb{F}}(J)\leqslant g_2$ for all $(b-a,w-v)$-evasive $\mathbb{F}$-subspace $J$ in $\mathbb{E}^{k-a}$. Then, it holds that $\dim_{\mathbb{F}}(U)\leqslant\max\{g_{1},g_{2}+v\}$.

{\bf{(2)}}\,\,Let $X$ be a $k$-dimensional $\mathbb{E}$-vector space, $(b,w)\in\mathbb{N}^{2}$ with $b\leqslant k$, and $U\leqslant_{\mathbb{F}}X$ be $(b,w)$-evasive in $X$. Moreover, let $t\in\{0,1,\dots,b\}$, and let $g_1\in\mathbb{R}$ such that $\dim_{\mathbb{F}}(I)\leqslant g_1$ for all $(b-t,w-t-1)$-evasive $\mathbb{F}$-subspace $I$ in $\mathbb{E}^{k}$. Then, it holds that $$\dim_{\mathbb{F}}(U)\leqslant\max\left\{g_{1},k-b+w,\frac{(k-b+t)m}{t+1}+w-t\right\}.$$
\end{proposition}

\begin{proof}
{\bf{(1)}}\,\,If $U$ is $(a,v-1)$-evasive in $X$, then we have $\dim_{\mathbb{F}}(U)\leqslant g_{1}$, as desired. Hence in what follows, suppose that there exists $A\leqslant_{\mathbb{E}}X$ such that $\dim_{\mathbb{E}}(A)=a$, $\dim_{\mathbb{F}}(U\cap A)\geqslant v$. This, together with the $(a,v)$-evasiveness of $U$, implies that $\dim_{\mathbb{F}}(U\cap A)=v$. By (3) of Lemma 2.3, $(U+A)/A$ is $(b-a,w-v)$-evasive in the $(k-a)$-dimensional $\mathbb{E}$-vector space $X/A$. It follows that $\dim_{\mathbb{F}}((U+A)/A)\leqslant g_{2}$, which further implies that $\dim_{\mathbb{F}}(U)\leqslant g_{2}+v$, as desired.

{\bf{(2)}}\,\,From (1) of Lemma 2.3, we deduce that $b\leqslant w$, and $U$ is $(b-t,w-t)$-evasive in $X$. Let $g_2\triangleq\max\{k-b+t,(k-b+t)m/(t+1)\}$. By (2) of Lemma 2.3, we infer that $\dim_{\mathbb{F}}(J)\leqslant g_2$ for all $(t,t)$-evasive $\mathbb{F}$-subspace $J$ in $\mathbb{E}^{k-b+t}$. Now an application of (1) to $(b-t,w-t)$ and $(b,w)$ yields that $\dim_{\mathbb{F}}(U)\leqslant \max\{g_{1},g_{2}+w-t\}$, which further implies the desired result.
\end{proof}

\setlength{\parindent}{2em}
\begin{remark}
If we set $w=b+1$, then (2) of Proposition 5.2 recovers [8, Lemma 3.11].
\end{remark}

Next, with the help of Proposition 5.2, we derive the following lemma.

\setlength{\parindent}{0em}
\begin{lemma}
{\bf{(1)}}\,\,Let $\lambda\in\mathbb{N}$, $u\in\mathbb{N}$. Moreover, let $k\in\mathbb{N}$ such that $k\geqslant\lambda$ and
\begin{equation}k^{2}-(m+\lambda-u-2)k\geqslant(\lambda-1)u+\lambda.\end{equation}
Then, it holds that $k\leftrightarrow(\lambda,0,u)$.

{\bf{(2)}}\,\,Let $\lambda\in\mathbb{N}$, $u\in\mathbb{N}$, and let $t\in\mathbb{N}$ such that $t^{2}+(\lambda+u+2-m)t\geqslant(m-1)\lambda-u$. Then, for any $a\in\mathbb{Z}^{+}$, it holds that $(\lambda+t,a-1,u+1)\curlyeqprec(\lambda,a,u)$.
\end{lemma}

\begin{proof}
{\bf{(1)}}\,\,Let $n\in\mathbb{N}$ with $n\geqslant k$, and let $U\leqslant_{\mathbb{F}}\mathbb{E}^{n}$ be $(n-\lambda,n-\lambda)$-evasive in $\mathbb{E}^{n}$. It suffices to show that $\dim_{\mathbb{F}}(U)\leqslant n+u$. By way of contradiction, suppose that $\dim_{\mathbb{F}}(U)\geqslant n+u+1$. Then, we have $\lambda\geqslant1$; moreover, from (2) of Lemma 2.3, we deduce that $n+u+1\leqslant mn/(n-\lambda+1)$, which further implies that $n^{2}-(m+\lambda-u-2)n\leqslant(\lambda-1)u+\lambda-1$. It then follows that $f(n)<0$ for $f\triangleq x^{2}-(m+\lambda-u-2)x-(\lambda-1)u-\lambda\in\mathbb{R}[x]$. Since $\lambda\geqslant1$, there exists $b,c\in\mathbb{R}$ such that $b<0<c$ and $f=(x-b)(x-c)$. By $f(n)<0$, we have $b<n<c$, which, together with $0\leqslant k\leqslant n$, implies that $b<k<c$, which further implies that $f(k)<0$, a contradiction to (5.4), as desired.

{\bf{(2)}}\,\,Fix $a\in\mathbb{Z}^{+}$. Let $k\in\mathbb{N}$ such that $k\leftrightarrow(\lambda+t,a-1,u+1)$, and we show that $k\leftrightarrow(\lambda,a,u)$. We note that $k\geqslant\lambda+t\geqslant\lambda$. Now let $n\in\mathbb{N}$ with $n\geqslant k$. Then, for any $J\leqslant_{\mathbb{F}}\mathbb{E}^{n}$ such that $J$ is $(n-\lambda-t,n-\lambda+a-t-1)$-evasive in $\mathbb{E}^{n}$, it holds that $\dim_{\mathbb{F}}(J)\leqslant n+(a-1)+(u+1)=n+a+u$. Let $U\leqslant_{\mathbb{F}}\mathbb{E}^{n}$ be $(n-\lambda,n-\lambda+a)$-evasive in $\mathbb{E}^{n}$. It suffices to show that $\dim_{\mathbb{F}}(U)\leqslant n+a+u$. Indeed, by $n\geqslant k\geqslant\lambda+t$, we have $t\in\{0,1,\dots,n-\lambda\}$. It then follows from Proposition 5.2 that
\begin{equation}\dim_{\mathbb{F}}(U)\leqslant\max\left\{n+a+u,n+a,\frac{(\lambda+t)m}{t+1}+n-\lambda+a-t\right\}.\end{equation}
Hence if $\dim_{\mathbb{F}}(U)\geqslant n+a+u+1$, then via some straightforward computation, (5.5) implies that $t^{2}+(\lambda+u+2-m)t\leqslant (m-1)\lambda-u-1$, a contradiction, as desired.
\end{proof}

\setlength{\parindent}{2em}
Now we state and prove the main result of this subsection.

\setlength{\parindent}{0em}
\begin{theorem}
{\bf{(1)}}\,\,Let $\lambda\in\mathbb{N}$, $u\in\mathbb{N}$, and let $(g_i\mid i\in\mathbb{Z}^{+})$ be a sequence of non-negative integers satisfy the following condition:
\begin{equation}\forall~s\in\mathbb{N}:{g_{s+1}}^{2}+\left(\lambda+u+2-m+\left(\sum_{i=1}^{s}g_{i}\right)+s\right)g_{s+1}\geqslant(m-1)\lambda+(m-1)\left(\sum_{i=1}^{s}g_{i}\right)-u-s.\end{equation}
Moreover, let $a\in\mathbb{N}$, and let $k\in\mathbb{N}$ such that $k\geqslant\lambda+\left(\sum_{i=1}^{a}g_{i}\right)$ and
\begin{equation}k^{2}-\left(m+\lambda-u+\left(\sum_{i=1}^{a}g_{i}\right)-a-2\right)k\geqslant\left(\lambda+\left(\sum_{i=1}^{a}g_{i}\right)-1\right)(u+a)+\lambda+\left(\sum_{i=1}^{a}g_{i}\right).\end{equation}
Then, it holds that $k\leftrightarrow\left(\lambda,a,u\right)$.

{\bf{(2)}}\,\,Let $\lambda\in\mathbb{N}$, $u\in\mathbb{N}$. Suppose that either $m=2$, or $(m=3$, $2u\geqslant\lambda)$, or $(m=4$, $u\geqslant\lambda+1)$. Then, for any $a\in\mathbb{N}$, it holds that $(a+\lambda+1)\leftrightarrow(\lambda,a,u)$.
\end{theorem}

\begin{proof}
{\bf{(1)}}\,\,First, for any $b\in\mathbb{Z}^{+}$ and $c\in\{0,1,\dots,b-1\}$, along with (5.6), an application of (2) of Lemma 5.2 to $\lambda+(\sum_{i=1}^{c}g_{i})\in\mathbb{N}$, $u+c\in\mathbb{N}$, $g_{c+1}\in\mathbb{N}$ and $b-c\in\mathbb{Z}^{+}$ immediately implies that \begin{equation}\left(\lambda+\left(\sum_{i=1}^{c+1}g_{i}\right),b-c-1,u+c+1\right)\curlyeqprec\left(\lambda+\left(\sum_{i=1}^{c}g_{i}\right),b-c,u+c\right).\end{equation}
Now by (5.7) and (1) of Lemma 5.2, we have $k\leftrightarrow\left(\lambda+\left(\sum_{i=1}^{a}g_{i}\right),0,u+a\right)$; moreover, by (5.8), we have $\left(\lambda+\left(\sum_{i=1}^{a}g_{i}\right),0,u+a\right)\curlyeqprec\left(\lambda,a,u\right)$, which further implies that $k\leftrightarrow(\lambda,a,u)$, as desired.

{\bf{(2)}}\,\,Define $(g_i\mid i\in\mathbb{Z}^{+})$ as $g_i=1$ for all $i\in\mathbb{Z}^{+}$. Via some straightforward computation, one can verify that $\lambda$, $u$ and $(g_i\mid i\in\mathbb{Z}^{+})$ satisfy (5.6); moreover, for any $a\in\mathbb{N}$, (5.7) holds true for $k=a+\lambda+1$. Hence the desired result immediately follows from (1).
\end{proof}

\setlength{\parindent}{2em}
We end this subsection with the following corollary, which will be used in Section 5.2.

\setlength{\parindent}{0em}
\begin{corollary}
Fix $(\lambda,s,k)\in\mathbb{N}^{3}$ with $k\geqslant s\geqslant\lambda+1$. Let $X$ be a $k$-dimensional vector space over $\mathbb{E}$, and let $J\leqslant_{\mathbb{F}}X$ be $(k-\lambda,2k-\lambda-s)$-evasive in $X$. Then, the following three statements hold:

{\bf{(1)}}\,\,If $m=2$, then $\dim_{\mathbb{F}}(J)\leqslant 2k-s$;

{\bf{(2)}}\,\,If $m=3$, then $\dim_{\mathbb{F}}(J)\leqslant 2k-s+\lceil\frac{\lambda}{2}\rceil$;

{\bf{(3)}}\,\,If $m=4$, then $\dim_{\mathbb{F}}(J)\leqslant 2k-s+\lambda+1$.
\end{corollary}

\begin{proof}
We only prove (2) and the proofs of (1), (3) are similar. By (2) of Theorem 5.1, we have $k-s+\lambda+1\leftrightarrow(\lambda,k-s,\lceil\frac{\lambda}{2}\rceil)$; moreover, $\dim_{\mathbb{E}}(X)=k\geqslant k-s+\lambda+1$, and $J$ is $(k-\lambda,k-\lambda+(k-s))$-evasive in $X$. It then follows that $\dim_{\mathbb{F}}(J)\leqslant k+(k-s)+\lceil\frac{\lambda}{2}\rceil=2k-s+\lceil\frac{\lambda}{2}\rceil$, as desired.
\end{proof}

\subsection{Lower bounds for $\varpi_{\mathbb{E}/\mathbb{F}}(k,r)$}

\setlength{\parindent}{2em}
Throughout this subsection, suppose that $m\geqslant2$. We will focus on the lower bounds for $\varpi_{\mathbb{E}/\mathbb{F}}(k,r)$. The following theorem is the first main result of this subsection, in which we show that the left hand side of (5.3) is not tight in general.

\setlength{\parindent}{0em}
\begin{theorem}
Let $r\in\mathbb{Z}^{+}$, and let $(a_i\mid i\in\mathbb{Z}^{+})$ be a sequence of non-negative integers satisfy the following condition:
\begin{equation}\forall~s\in\mathbb{N}:{a_{s+1}}^{2}+\left(m(r-1)+2+\left(\sum_{i=1}^{s}a_i\right)+s\right)a_{s+1}\geqslant(m-1)\left(\sum_{i=1}^{s}a_i\right)+m-s.\end{equation}
Moreover, let $w\in\mathbb{N}$, and let $k\in\mathbb{N}$ such that $k\geqslant r+1+\left(\sum_{i=1}^{w}a_i\right)$, $(m-1)(k-r)\geqslant w$ and
\begin{equation}k^{2}-\left(2r-m(r-1)+\left(\sum_{i=1}^{w}a_i\right)-w\right)k\geqslant\left(\left(\sum_{i=1}^{w}a_i\right)+r\right)((m-1)r+w)+1.\end{equation}
Then, for any $\varepsilon\in\mathbb{N}$ with $\varepsilon\geqslant k$, it holds that $\varpi_{\mathbb{E}/\mathbb{F}}(\varepsilon,r)\geqslant (m-1)r+\varepsilon+w+1$.
\end{theorem}

\begin{proof}
Set $\lambda=r+1$, $u=(m-1)r-1$, $(g_i\mid i\in\mathbb{Z}^{+})=(a_i\mid i\in\mathbb{Z}^{+})$ and $a=w$ in Theorem 5.1. From (5.9) and (5.10), we deduce that $k\leftrightarrow(r+1,w,(m-1)r-1)$. Let $\varepsilon\in\mathbb{N}$ with $\varepsilon\geqslant k$. By $(m-1)(k-r)\geqslant w$, we have $(m-1)r+\varepsilon+w\leqslant\varepsilon m$; moreover, for any $(\varepsilon-r-1,\varepsilon-r-1+w)$-evasive $\mathbb{F}$-subspace $J$ in $\mathbb{E}^{\varepsilon}$, it holds that $\dim_{\mathbb{F}}(J)\leqslant \varepsilon+w+(m-1)r-1$. Hence an application of Corollary 4.1 to $b=(m-1)r+\varepsilon+w$ immediately implies the desired result.
\end{proof}

\setlength{\parindent}{2em}
Now we consider some special cases of Theorem 5.2. Part (1) of the following corollary is exactly (2) of Lemma 5.1 when $\mathbb{E}$ is finite, which we recover among other results.

\setlength{\parindent}{0em}
\begin{corollary}
{\bf{(1)}}\,\,For any $k\in\mathbb{N}$ with $k\geqslant\lfloor\sqrt{m}\rfloor+2$, it holds that $\varpi_{\mathbb{E}/\mathbb{F}}(k,1)\geqslant m+k$.

{\bf{(2)}}\,\,For any $k\in\mathbb{N}$ such that $k^{2}-\lceil\sqrt{m+1}\rceil k\geqslant\lceil\sqrt{m+1}\rceil m+1$, it holds that $\varpi_{\mathbb{E}/\mathbb{F}}(k,1)\geqslant m+k+1$.

{\bf{(3)}}\,\,For any $(k,r)\in\mathbb{N}^{2}$ such that $k\geqslant r+2$, $r\geqslant2$, it holds that $\varpi_{\mathbb{E}/\mathbb{F}}(k,r)\geqslant (m-1)r+k+1$.

{\bf{(4)}}\,\,For any $(k,r)\in\mathbb{N}^{2}$ such that $k\geqslant r+3$, $r\geqslant3$, it holds that $\varpi_{\mathbb{E}/\mathbb{F}}(k,r)\geqslant (m-1)r+k+2$;

{\bf{(5)}}\,\,For $k\in\mathbb{N}$ such that either $k\geqslant 6$ or $(k=5,m\leqslant7)$, it holds that $\varpi_{\mathbb{E}/\mathbb{F}}(k,2)\geqslant 2m+k$.
\end{corollary}

\begin{proof}
All the statements follow from Theorem 5.2 via some verification. More specifically, (1) follows from setting $r=1$, $w=0$ and $k=\lfloor\sqrt{m}\rfloor+2$ in Theorem 5.2, (2) follows from setting $r=1$, $s=0$, $a_1=\lfloor\sqrt{m}\rfloor$ and $w=1$ in Theorem 5.2, (3) follows from setting $w=0$ and $k=r+2$ in Theorem 5.2, (4) follows from setting $s=0$, $a_1=1$, $w=1$ and $k=r+3$ in Theorem 5.2, and (5) follows from setting $r=2$, $s=0$, $a_1=1$, $w=1$ and $k=5,6$ in Theorem 5.2, respectively.
\end{proof}

\setlength{\parindent}{2em}
Now we further improve the left hand side of (5.3) for small $m$. The following theorem is the second main result of this subsection.

\setlength{\parindent}{0em}
\begin{theorem}
{\bf{(1)}}\,\,If $m=3$, then for any $k\in\mathbb{N}$ with $k\geqslant 2$, it holds that $\varpi_{\mathbb{E}/\mathbb{F}}(k,1)\geqslant2k$.

{\bf{(2)}}\,\,If $m=4$, then for any $k\in\mathbb{N}$ with $k\geqslant 3$, it holds that $\varpi_{\mathbb{E}/\mathbb{F}}(k,2)\geqslant2k+3$.
\end{theorem}

\begin{proof}
We only prove (1) and the proof of (2) is similar. If $k=2$, then (3) of Corollary 5.1 implies that $\varpi_{\mathbb{E}/\mathbb{F}}(k,1)=2k$, as desired. Therefore in what follows, suppose that $k\geqslant 3$. By (2) of Corollary 5.2, for any $(k-2,2k-5)$-evasive $\mathbb{F}$-subspace $J$ in $\mathbb{E}^{[k]}$, it holds that $\dim_{\mathbb{F}}(J)\leqslant 2k-2$. Hence an application of Corollary 4.1 to $r=1$ and $b=2k-1$ immediately implies the desired result.
\end{proof}

\setlength{\parindent}{2em}
\begin{remark}
In Section 5.3, we will show that if $\mathbb{E}$ is finite and $m=3$, then $\varpi_{\mathbb{E}/\mathbb{F}}(k,1)=2k$.
\end{remark}

\subsection{An upper bound for $\varpi_{\mathbb{E}/\mathbb{F}}(k,r)$}
\setlength{\parindent}{2em}
In this subsection, suppose that $\mathbb{F}$ is finite with $|\mathbb{F}|=q$. We will use Theorem 3.4 to derive a general upper bound for $\varpi_{\mathbb{E}/\mathbb{F}}(k,r)$. As a first step, we compute the number of all the $[n,r+1]$ $r$-minimal codes, as detailed in the following proposition.

\setlength{\parindent}{0em}
\begin{proposition}
Let $(n,r)\in\mathbb{N}^{2}$ with $n\geqslant r+1$. Then, it holds that:
\begin{eqnarray*}
\begin{split}
|\{C\mid\text{$C$ is an $[n,r+1]$ $r$-minimal code}\}|=\left(\sum_{i=mr+1}^{m(r+1)}\delta_{q}(m(r+1),i)\bin_{q}(n,i)\right)\delta_{q^{m}}(r+1,r+1)^{-1}.
\end{split}
\end{eqnarray*}
\end{proposition}

\begin{proof}
For any $[n,r+1]$ code $C\leqslant_{\mathbb{E}}\mathbb{E}^{n}$, by Theorem 4.4 and Corollary 2.1, $C$ is $r$-minimal if and only if $mr+1\leqslant\wt(C)\leqslant m(r+1)$. Therefore the number of all the $[n,r+1]$ $r$-minimal codes is equal to $\sum_{i=mr+1}^{m(r+1)}\theta_{i}$, where $\theta_{i}$ denotes the number of all the $[n,r+1]$ codes with rank support weight $i$. Now we compute $\theta_{i}$ for any fixed $i\in\{mr+1,\dots,m(r+1)\}$. Noticing that each $[n,r+1]$ code has $\delta_{q^{m}}(r+1,r+1)$ generator matrices in $\Mat_{r+1,n}(\mathbb{E})$, all of which have rank $r+1$; moreover, any $G\in\Mat_{r+1,n}(\mathbb{E})$ with $\rank(G)=r+1$ is a generator matrix of the $[n,r+1]$ code $C\triangleq\langle \row(G)\rangle_{\mathbb{E}}$, and it holds that $\wt(C)=\dim_{\mathbb{F}}(\langle \col(G)\rangle_{\mathbb{F}})$ by Lemma 2.1. It then follows that
$$\theta_i=\delta_{q^{m}}(r+1,r+1)^{-1}\cdot|\{G\in\Mat_{r+1,n}(\mathbb{E})\mid\text{$\rank(G)=r+1$, $\dim_{\mathbb{F}}(\langle \col(G)\rangle_{\mathbb{F}})=i$}\}|.$$
We claim that every $G\in\Mat_{r+1,n}(\mathbb{E})$ with $\dim_{\mathbb{F}}(\langle \col(G)\rangle_{\mathbb{F}})=i$ satisfies that $\rank(G)=r+1$. Indeed, from $\dim_{\mathbb{F}}(\langle \col(G)\rangle_{\mathbb{E}})\geqslant i>mr$, we deduce that $\rank(G)=\dim_{\mathbb{E}}(\langle \col(G)\rangle_{\mathbb{E}})>r$, which implies that $\rank(G)=r+1$, as desired. Next, let $K\in\Mat_{r+1,m(r+1)}(\mathbb{E})$ such that $\col(K)$ form an $\mathbb{F}$-basis of $\mathbb{E}^{[r+1]}$. Then, $H\mapsto KH$ induces a bijection from $\Mat_{m(r+1),n}(\mathbb{F})$ to $\Mat_{r+1,n}(\mathbb{E})$; moreover, we have $\rank(H)=\dim_{\mathbb{F}}(\langle\col(KH)\rangle_{\mathbb{F}})$ for all $H\in\Mat_{m(r+1),n}(\mathbb{F})$. It follows that
\begin{eqnarray*}
\begin{split}
\theta_i&=\delta_{q^{m}}(r+1,r+1)^{-1}\cdot|\{H\in\Mat_{m(r+1),n}(\mathbb{F})\mid\rank(H)=i\}|\\
&=\delta_{q^{m}}(r+1,r+1)^{-1}\delta_{q}(m(r+1),i)\bin_{q}(n,i),
\end{split}
\end{eqnarray*}
which immediately establishes the desired result.
\end{proof}

\setlength{\parindent}{2em}
Now we are ready to prove the main result of this subsection.

\setlength{\parindent}{0em}
\begin{theorem}
Let $(k,r)\in\mathbb{N}^{2}$ with $k\geqslant r+1$. Then, there exists an $[mr+k(r+1)-r^{2}-2r,k]$ $r$-minimal code, i.e., it holds that $\varpi_{\mathbb{E}/\mathbb{F}}(k,r)\leqslant mr+k(r+1)-r^{2}-2r$.
\end{theorem}

\begin{proof}
Let $n\triangleq mr+k(r+1)-r^{2}-2r$. We note that if $r=0$, then (3) of Corollary 5.1 implies that $\varpi_{\mathbb{E}/\mathbb{F}}(k,r)=n$, as desired; moreover, if $r+1\geqslant m$, then (2) of Corollary 5.1 implies that $\varpi_{\mathbb{E}/\mathbb{F}}(k,r)=(k-1)m+1\leqslant n$, as desired. Therefore in what follows, we assume that $m\geqslant r+2\geqslant3$. We first show that
\begin{equation}
\sum_{i=0}^{mr}\delta_{q}(m(r+1),i)\bin_{q}(n,i)<q^{m(n-k)(r+1)}\delta_{q^{m}}(r+1,r+1).
\end{equation}
Indeed, from (2) of Proposition 2.2, we deduce that
{\small$$q^{-mr(m+n)}\left(\sum_{i=0}^{mr}\delta_{q}(m(r+1),i)\bin_{q}(n,i)\right)<\frac{1}{(q-1)(q^{2}-q-1)q^{(k-r-1)(r+1)+m-1}}+\frac{q^{2}}{q^{2}-q-1}\leqslant\frac{17}{4},$$}where the second ``$\leqslant$'' follows from $k\geqslant r+1$, $m\geqslant3$ and $q\geqslant2$. Moreover, noticing that $q^{m}\geqslant8$, from (1) of Proposition 2.2, (2.5) and some straightforward computation, we deduce that
$$q^{-mr(m+n)}\cdot q^{m(n-k)(r+1)}\delta_{q^{m}}(r+1,r+1)=q^{m}\left(\prod_{i=1}^{r+1}(1-q^{-mi})\right)>q^{m}-1-q^{-m}\geqslant\frac{55}{8},$$
which, together with $\frac{17}{4}<\frac{55}{8}$, immediately implies (5.11), as desired. Next, we show that
\begin{equation}
\delta_{q^{m}}(n,r+1)-\left(\sum_{i=mr+1}^{m(r+1)}\delta_{q}(m(r+1),i)\bin_{q}(n,i)\right)<\left(\prod_{i=k-r}^{k}\frac{q^{m(i+n-k)}-1}{q^{mi}-1}\right)\delta_{q^{m}}(r+1,r+1).
\end{equation}
Indeed, by (2.5) and [18, Lemma 26], we have
$$\delta_{q^{m}}(n,r+1)<q^{m(r+1)n}=\sum_{i=0}^{m(r+1)}\delta_{q}(m(r+1),i)\bin_{q}(n,i),$$
which implies that the left hand side of (5.12) is smaller than that of (5.11). Moreover, from $n\geqslant k$, we deduce that the right hand side of (5.11) is smaller than or equal to that of (5.12). Therefore (5.12) follows from (5.11), as desired. Dividing both sides of (5.12) by $\delta_{q^{m}}(r+1,r+1)$, we have $$\bin_{q^{m}}(n,r+1)-\delta_{q^{m}}(r+1,r+1)^{-1}\left(\sum_{i=mr+1}^{m(r+1)}\delta_{q}(m(r+1),i)\bin_{q}(n,i)\right)<\prod_{i=k-r}^{k}\frac{q^{m(i+n-k)}-1}{q^{mi}-1}.$$
By Proposition 5.3, the left hand side is equal to the number of all the $[n,r+1]$ codes which are not $r$-minimal. Hence an application of Theorem 3.4 immediately implies the desired result.
\end{proof}

\setlength{\parindent}{2em}
Along with (1) of Theorem 5.3, an application of Theorem 5.4 to $r=1$ immediately implies the following result for minimal codes.

\setlength{\parindent}{0em}
\begin{corollary}
{\bf{(1)}}\,\,For any $k\in\mathbb{N}$ with $k\geqslant2$, it holds that $\varpi_{\mathbb{E}/\mathbb{F}}(k,1)\leqslant m+2k-3$.

{\bf{(2)}}\,\,If $m=3$, then for any $k\in\mathbb{N}$ with $k\geqslant2$, it holds that $\varpi_{\mathbb{E}/\mathbb{F}}(k,1)=2k$.
\end{corollary}

\begin{remark}
{\bf{(1)}}\,\,Part (1) of Corollary 5.4 slightly improves [3, Theorem 6.11], and also recovers [26, Theorem 7.16] when $k=3$ (see (4), (5) of Lemma 5.1). Moreover, the upper bound $m+2k-3$ is not tight in general, and we refer the reader to (6)--(8) of Lemma 5.1 and references therein for more details.

{\bf{(2)}}\,\,If $m=3$, then (2) of Corollary 5.4 implies that $\varpi_{\mathbb{E}/\mathbb{F}}(4,1)=8$, which recovers the main result of [8, Section 3.4] (see (3) of Lemma 5.1). Moreover, an explicit construction of $[8,4]$ minimal codes has been given in \cite{8}, and we refer the reader to [8, Proposition 3.13] for more details.
\end{remark}

\end{document}